\documentclass{llncs}
\pdfoutput=1
\usepackage{amssymb}
\usepackage{xspace}
\usepackage{tikz}
\usepackage{amsmath}
\usepackage{multirow}
\usepackage{scrextend}
\usepackage{wrapfig}
\usepackage{cite}
\usepackage{hyperref}

\newtheorem{algorithm} {Algorithm}

\renewcommand{\iff}{\Leftrightarrow}

\newcommand{\mcomment}[1]{}

\newcommand{\A}{\ensuremath{\mathcal{A}}\xspace}

\newcommand{\D}{\ensuremath{\mathcal{D}}\xspace}

\newcommand {\union} {\ensuremath{\cup}}

\renewcommand {\iff} {\ensuremath{\Leftrightarrow}}

\newcommand {\qddc} {QDDC\xspace}

\newcommand {\ang}[1] {\ensuremath{\langle#1\rangle}}
\newcommand {\sq}[1] {\ensuremath{[#1]}}
\newcommand {\dsq}[1] {\ensuremath{[[#1]]}}
\newcommand {\dcurly}[1] {\ensuremath{\{\{#1\}\}}}

\newcommand {\len}[1] {\ensuremath{len(#1)}}

\newcommand {\intv}[1] {\ensuremath{Intv(#1)}}
\newcommand {\slen} {\ensuremath{slen} }
\newcommand {\scount} {\ensuremath{scount} }
\newcommand {\sdur} {\ensuremath{sdur} }

\newcommand {\true} {\ensuremath{true}~}

\newcommand {\dom}[1] {\ensuremath{dom(#1)}}

\newcommand {\nat} {\ensuremath{\mathbb{N}}}

\newcommand{\oomit}[1]{}

\newcommand{\df}{=}

\newcommand{\DCSYNTH}{DCSynth\xspace}
\newcommand{\invariant}{\mbox{\bf inv\/}}
\newcommand{\MPNC}{MPS\xspace}
\newcommand{\GODSC}{MPHOS\xspace}

\newcommand{\TYPE}{Type\xspace}

\newcommand{\mpnc}{mps\xspace}
\newcommand{\godsc}{mphos\xspace}

\begin{document}
\title{\DCSYNTH: Guided Reactive Synthesis with Soft Requirements}
\author{{Amol Wakankar$^1$} \and {Paritosh K. Pandya$^2$} \and {Raj Mohan Matteplackel$^2$}
}
 
\institute{
Homi Bhabha National Institute, Mumbai, India.\\
Bhabha Atomic Research Centre, Mumbai, India.\\
\and
Tata Institute of Fundamental Research, Mumbai 400005, India.\\
}


\maketitle

\setcounter{footnote}{0}

\begin{abstract}
In reactive controller synthesis, 
 a number of implementations (controllers) are possible for a given specification because of incomplete nature of specification. To choose the most desirable one from the various options, we need to specify additional properties which can guide the synthesis. 
 In this paper, We propose a technique for guided controller synthesis from regular requirements which are specified  using an interval temporal logic \qddc.
We find that \qddc is well suited for guided synthesis due to its superiority in dealing with both qualitative and quantitative specifications. 
Our framework allows specification consisting of both {\em hard} and {\em soft} requirements as \qddc formulas.

We have also developed a method and a tool \DCSYNTH, which computes a controller that \emph{invariantly} satisfies the \emph{hard requirement} and it {\em optimally} meets the {\em soft requirement}. 
The proposed technique is also useful in dealing with  conflicting i.e., unrealizable requirements, by making some of the them as \emph{soft} requirements. 
Case studies are carried out to demonstrate the effectiveness of the soft requirement guided synthesis in obtaining high quality controllers. The quality of the synthesized controllers is  compared using metrics measuring both the \emph{guaranteed} and the \emph{expected case} behaviour of the controlled system. Tool \DCSYNTH facilitates such comparison.
\end{abstract}

\section{Introduction}
\label{section:intro}
\vspace{-1ex}
Reactive synthesis aims at constructing a controller (say a Mealy Machine) algorithmically from a given temporal logic specification of its desired behaviour.
Considerable amount of research has gone into the area of reactive synthesis
and several tools are available for experimenting with reactive synthesis\cite{SYN17}. However, existing tools do not have the capability to guide the synthesis towards the most desirable controller. In practice, user specification may be incomplete and it may contain certain requirements, which are not \textit{mandatory}, but desirable. We term the desirable properties as \textit{soft} requirements.

In this work, we propose a specification consisting of \emph{hard requirement} which are mandatory and needs to be satisfied \emph{invariantly}, and the \emph{soft requirement} which are desirable and should be satisfied at as many points in the execution as possible. We choose to specify the hard and soft requirements as regular properties in logic \emph{Quantified Discrete Duration Calculus}(\textbf{\qddc})\cite{Pan01b,Pan01a}.
%
%
%
%
\qddc is the discrete time variant of Duration Calculus proposed by Zhou et.al. \cite{CHR91,ZH04}.

Regular properties can conceptually be specified by a deterministic finite state automaton (DFA). At any point in the execution, a regular property holds provided the past behaviour upto the point is accepted by its DFA.  
The study of synthesis of controllers for such properties was pioneered by Ramadge and Wonham \cite{RW87,RW89,LRT17}. 
\qddc\/ is an interval temporal logic, which has the expressive power of regular languages. Section \ref{section:qddc} presents the syntax and semantics of this logic. Prior work \cite{SP05,Pan01a,Pan01b} shows that any formula in \qddc\/ can be effectively translated into a language equivalent DFA over finite words. Logic \qddc's 
bounded counting features, interval based modalities and regular expression like primitives allow complex qualitative and quantitative  properties (such as latency, resource constraints) to  be specified succinctly and modularly
(see the example below). 
With illustrations, papers \cite{MPW17, Pan01a, Pan02} show how quantitative and qualitative regular properties can be succinctly  specified in \qddc. Paper \cite{MPW17} also gives a comparison with other logics such as LTL and PSL. It should be noted that
\qddc does not allow specification of general liveness properties, however, time bounded liveness can be specified. The following example motivates need for soft requirement guided synthesis.
\vspace*{-1ex}
\begin{example}[Arbiter for Mutually Exclusive Shared Resource]
\label{exam:arbiter}
The arbiter has an input $r_i$ (denoting request for access) and an output $a_i$ (denoting acknowledgement for access) for each client $1 \leq i \leq n$. The specification consists of the following two 
properties, given as \qddc\/ formulas together with their intuitive explanation. Section
\ref{section:qddc} gives the formal syntax and semantics of \qddc.

\noindent $-$\textit{Mutual Exclusion Requirement} $R_1$: $[[~\land_{i \neq j} ~\neg ( a_i \land a_j) ~]]$, states that at every point, the access to the shared resource should be mutually exclusive.

\noindent $-$\textit{k-cycle Response Requirement} $R_2$: $[](\land_i ~(([[r_i]] ~\&\&\ (slen>=(k-1)))~\Rightarrow (scount ~a_i > 0))$ , states that in any observation  interval spanning $k$ cycles if request from $i^{th}$ client ($r_i$) is \textit{continuously} high during the interval, then  that client should get at least one access ($a_i$)  within the observation interval. Modality $[]D$ states that  sub-formula  $D$ should hold for all observation intervals. Term $slen$ gives the length of the observation interval and the term $(scount~ P)$ counts the number of occurrences of proposition $P$ within the observation interval. The property $R2$ is asserted for each client $1 \leq i \leq n$.  

When $k < n$ no controller can satisfy both requirements (their conjunction is unrealizable); consider the case where all clients request all the time.
We may want to opt for an implementation, which mandatorily satisfies $R1$ and 
it tries to meet the $R2$ ``as much as possible''. This can be specified in our framework, by making requirement $R2$ as a \emph{soft requirement}. One possible controller in such a case will make $a_i$ true in a \textit{round-robin} manner for all the requesting clients. Let $r$ denote the number of clients
requesting simultaneously. The round-robin policy will mandatorily satisfy $R2$
under the assumption that invariantly $r \leq k$. But when  $r > k$ (\textit{overload condition}), this controller will be able to meet the response time requirement $R2$ for only $k$ out of $n$ clients. Using the soft requirement, we can also given priority to the clients which meet response time requirements under \textit{overload condition}. Thus, soft requirements  allow us to choose the preferred one from several candidate controllers. \qed
\end{example}

This paper introduces a tool {\em \DCSYNTH} which allows synthesis of controllers 
from {\em regular properties} (\qddc formulas). 
%
The specification in \emph{\DCSYNTH\/} is a tuple $(I, O,D^h,D^s)$, where $D^h$ and $D^s$ are \qddc\/ formulas over a set of input and output propositions $(I,O)$. Here,
$D^h$ and $D^s$ are the {\bf hard} and the {\bf soft} requirement, respectively\footnote{The tool supports more general lexicographically ordered list of soft requirements. However, we omit the general case for brevity}.
We use the term {\em supervisor} for a non-blocking Mealy machine which may non-deterministically produce one or more outputs for each input. A supervisor may be
refined to a sub-supervisor by resolving (pruning) the non-determinstic choice of
outputs (the sub-supervisor may use additional memory for making the choice.) We define
a determinism ordering on supervisors in the paper. A \emph{controller} is a deterministic supervisor.
Ramadge and Wonham  \cite{RW87,RW89} investigated the synthesis of the {\em maximally permissive}  supervisor for a regular specification. The maximally permissive supervisor 
is a unique supervisor, which encompasses all the behaviors invariantly satisfying the specified regular property (See Definition \ref{def:realizableAndMPNC}). 
The well known safety synthesis algorithm applied to the DFA for $D^h$ gives us the maximally permissive supervisor $\MPNC(D^h)$ \cite{ELTV17}. 
If no such supervisor 
exists, the specification is reported as unrealizable. 

 
Any controller obtained by arbitrarily resolving the nondeterministic choices for outputs in $\MPNC(D^h)$ is correct-by-construction. This results in several controllers with distinct behaviours (as shown by previous example). Thus, only correct-by-construction synthesis is not sufficient \cite{BCGHHJKK14}. Some form of  guidance must be provided to the synthesis method to choose 
among the possible controllers. 
We use the soft requirements to provide such guidance. Our synthesis method tries to choose a controller, which satisfies the soft requirements ($D^s$) ``as much as possible''. Soft requirement can also specify the desirable requirements, which cannot be met invariantly. For example, in a Mine-pump controller, as soft requirement ``keep the pump off unless mandated by the hard requirement'' specifies an energy efficient controller. Specification of scheduling, performance and quality constraints are often such desirable properties. 
Moreover, a specification may consist of a conjunction of  conflicting requirements. In this case, all the requirements cannot be invariantly met simultaneously. User may resolve the conflict(s) by making some of these requirements as soft. Therefore, soft requirements give us a capability to synthesize meaningful and practical controllers.

In \DCSYNTH, we formalize the notion of a controller meeting the soft requirement $D^s$ ``as much as possible'', by synthesizing a sub-supervisor of $\MPNC(D^h)$  (guaranteeing invariance of $D^h$), which maximizes 
the expected value of count of $D^s$ in next $H$ moves when averaged over all the inputs.
The classical value iteration  algorithm due to Bellman \cite{Bel57} allows us to compute this $H$-optimal sub-supervisor. This can be further refined to a controller as desired.  {\em Thus, our synthesis method gives a controller which, (a) invariantly
satisfies $D^h$ and (b) it is $H$-optimal for $D^s$ amongst all controllers meeting condition (a).}

The above synthesis method is implemented in tool \DCSYNTH. An efficient representation of DFA using BDDs, originally introduced by the tool MONA \cite{Mon01}, is used for representing both automata and supervisors.
We adapt the safety synthesis algorithm and the value iteration algorithm so that they work symbolically over this MONA DFA representation.
\oomit{ Moreover, being in the realm of regular properties, we are able to minimize automata and supervisors, giving significantly efficient 
synthesis and small controllers. The full paper gives detailed algorithms and the implementation details (See Appendix \ref{sec:ssdfa}) \cite{WPM18}, which are omitted here because of the lack of space.}

We illustrate our specification method and synthesis tool with the help of two case studies\footnote{\DCSYNTH\/ can be downloaded at \cite{WPM18} along with the specification files for experiments.}. 
We define metrics to compare the controllers for their  \textit{guaranteed} and \textit{expected behaviour}. The tool \DCSYNTH  facilitates measurement of both these metrics.
The main contributions of this paper are as follows:
\begin{itemize}
 \item We develop a technique for the synthesis of controllers from \qddc\/ requirements. 
This extends the past work on model checking interval temporal logic \qddc\cite{Pan01a,Pan01b,CP03,SPC05,SP05} with synthesis abilities. 
 \item We propose a method for 
 guided synthesis of controllers based on {\bf soft requirements} which are met in a $H$-optimal fashion. 
Conceptually, this enhances the  Ramadge-Wonham framework for optimal controller synthesis.
 \item We present a tool \DCSYNTH\/ for guided synthesis, which 
   \begin{itemize}
       \item represents and manipulates automata/supervisors using BDD-based  semi-symbolic DFA. It uses eager minimization for efficient synthesis, and
       \item provides facility to compare both the guaranteed and expected case behaviours of the candidate controllers. 
   \end{itemize}
 \item  We analyse the impact of soft requirements on the quality of the synthesized controllers experimentally using case studies.
\end{itemize}

The rest of the paper is arranged as follows. Section \ref{section:qddc} describes the syntax and semantics of \qddc. Important definitions are presented in Section \ref{sec:definitions}. Syntax of \DCSYNTH specification and the controller synthesis method are presented in Section \ref{section:dcsynth-spec}. 
Section \ref{section:motivation} discusses
case studies and experimental results. 
The paper is concluded with a discussion and related work in Sections \ref{section:discussion} and \ref{section:conclusion}. 

\section{Quantified Discrete Duration Calculus (\qddc) Logic}
\label{section:qddc}
Let $PV$ be a finite non-empty set of propositional variables. 
Let $\sigma$ a non-empty finite word over the alphabet
$2^{PV}$. It has the form 
$\sigma=P_0\cdots P_n$ where $P_i\subseteq PV$ for each $i\in\{0,\ldots,n\}$. 
Let $\len{\sigma}=n+1$, 
$\dom{\sigma}=\{0,\ldots,n\}$, 
$\sigma[i,j]=P_i \cdots P_j$ and $\sigma[i]=P_i$.

The syntax of a \emph{propositional formula} over variables $PV$ is given by:
\vspace{-1ex}
\[
\varphi := false\ |\ \true\ |\ p\in PV |\ !\varphi\ |\ \varphi~\&\&~\varphi\ |\ \varphi~||~\varphi
\]
with $\&\&, ||, !$ denoting conjunction, dis-junction and negation, respectively. Operators
such as $\Rightarrow$  and $\Leftrightarrow$ are defined as usual. 
Let $\Omega(PV)$ be the set of all propositional formulas over variables $PV$. 
%
Let $i\in\dom{\sigma}$. 
Then the satisfaction of propositional formula $\varphi$ at point $i$, denoted $\sigma,i\models\varphi$ is defined as usual and omitted here for brevity.
\oomit{
$\forall i\in\dom{\sigma}: \sigma, i \models \true$, $\sigma, i \models p$ iff $p\in\sigma[i]$, 
and $\sigma,i \models - \varphi$ iff $i>0$ and $\sigma, i-1 \models \varphi$. 
The satisfaction of
boolean combinations \verb#!# (not), \verb#&&# (and), \verb#||# (or) is  defined in a natural way.
}

The syntax of a \qddc formula over variables $PV$ is given by: 
\[
\begin{array}{lc}
D:= &\ang{\varphi}\ |\ \sq{\varphi}\ |\ \dsq{\varphi}\ |\ 
D\ \verb|^|\ D\ |\ !D\ |\ D~||~D\ |\ D~\&\&~D\ \\ 
&ex~p.\ D\ |\ all~p.\ D\ |\ slen \bowtie c\ |\ scount\ \varphi \bowtie c\ |\ sdur\ \varphi \bowtie c 
\end{array} 
\]
where $\varphi\in\Omega(PV)$, $p\in PV$, 
$c ~ \in\nat$ and $\bowtie\in\{<,\leq,=,\geq,>\}$. 

An \emph{interval} over a word $\sigma$ is of the form $[b,e]$ 
where $b,e\in\dom{\sigma}$ and $b\leq e$. 
Let $\intv{\sigma}$ be the set of all intervals over $\sigma$.
Let $\sigma$ be a word over $2^{PV}$ and let $[b,e]\in\intv{\sigma}$ be an interval. 
Then the satisfaction relation of a \qddc formula $D$ over
$\Sigma$ and interval $[b,e]$ written as $\sigma,[b,e]\models D$, is defined inductively as follows:
\[
\begin{array}{lcl}
\sigma, [b,e]\models\ang{\varphi} & \mathrm{\ iff \ } & b=e \mbox{ and } \sigma,b\models \varphi,\\
\sigma, [b,e]\models\sq{\varphi} & \mathrm{\ iff \ } & b<e \mbox{ and }
                                           \forall b\leq i<e:\sigma,i\models \varphi,\\
\sigma, [b,e]\models\dsq{\varphi} & \mathrm{\ iff \ } & \forall b\leq i\leq e:\sigma,i\models \varphi,\\
\sigma, [b,e]\models\dcurly{\varphi} & \mathrm{\ iff \ } & e=b+1 \mbox{ and }\sigma,b\models \varphi,\\
\sigma, [b,e]\models D_1\verb|^| D_2 & \mathrm{\ iff \ } & \exists b\leq i\leq e:\sigma, [b,i]\models D_1\mbox{ and }\sigma,[i,e]\models D_2,\\
\end{array}
\]
with Boolean combinations $!D$, $D_1~||~D_2$ and $D_1~\&\&~D_2$  defined in the expected way. 
We call word $\sigma'$ a $p$-variant, $p\in PV$, of a word $\sigma$ 
if $\forall i\in\dom{\sigma},\forall q\neq p:q\in \sigma'[i]\iff q\in \sigma[i]$. 
Then $\sigma,[b,e]\models ex~p.~D\iff\sigma',[b,e]\models D$ for some 
$p$-variant $\sigma'$ of $\sigma$ and 
$(all~p.~D) \Leftrightarrow (!ex~p.~!D)$. 
%

%


Entities \slen, \scount and \sdur are called \emph{terms}. 
The term \slen gives the length of the interval in which it is 
measured, $\scount\ \varphi$ where $\varphi\in\Omega(PV)$, counts 
the number of positions including the last point 
in the interval under consideration where $\varphi$ holds, and    
$\sdur\ \varphi$ gives the number of positions excluding the last point 
in the interval where $\varphi$ holds. 
Formally, for $\varphi\in\Omega(PV)$ we have 
$\slen(\sigma, [b,e])=e-b$, $\scount(\sigma,\varphi,[b,e])=\sum_{i=b}^{i=e}\left\{\begin{array}{ll}
					1,&\mbox{if }\sigma,i\models\varphi,\\
					0,&\mbox{otherwise.}
					\end{array}\right\}$ and 
$\sdur(\sigma,\varphi,[b,e])=\sum_{i=b}^{i=e-1}\left\{\begin{array}{ll}
					1,&\mbox{if }\sigma,i\models\varphi,\\
					0,&\mbox{otherwise.}
					\end{array}\right\}$

We also define the following derived constructs: 
$pt=\langle true \rangle$, $ext=!pt$, $\mathbf{ \langle \rangle D} = true\verb|^|D\verb|^|true$, $[]D=(!\langle \rangle!D)$ and $\mathbf{pref(D)}=!((!D)\verb|^|true)$. 
Thus, $\sigma, [b,e] \models []D$ iff $\sigma, [b',e']\models D$ 
for all sub-intervals $b\leq b'\leq e'\leq e$ and $\sigma, [b,e] \models \mathit{pref}(D)$ 
iff $\sigma, [b,e']\models D$ for all prefix intervals $b \leq e' \leq e$.

\begin{definition}[Language of a formula]
\label{def:dclang}
Let $\sigma,i \models D$ iff $\sigma, [0,i] \models D$, and 
$\sigma \models D$ iff $\sigma, \len{\sigma}-1 \models D$. We define
$L(D) = \{ \sigma\mid\sigma \models D \}$, the set of behaviours accepted by $D$. Formula $D$ is called valid, denoted $\models_{dc} D$, iff $L(D) = (2^{PV})^+$. \qed
\end{definition}
Thus, a formula $D$ holds at a point $i$ in a behaviour provided the {\bf past} of the point $i$ satisfies $D$.
\begin{theorem}
\label{theorem:formula-automaton}
 \cite{Pan01a} For every formula $D$ over variables $PV$ we can construct a Deterministic Finite Automaton (DFA) $\A(D)$ over alphabet $2^{PV}$ 
such that $L(\A(D))=L(D)$. We call $\A(D)$ a \emph{formula automaton} for $D$ or the monitor automaton for $D$. \qed
\end{theorem}
A tool DCVALID implements this formula automaton construction in an efficient manner by internally using the tool MONA \cite{Mon01}. 
It gives {\em minimal, deterministic} automaton (DFA) for the formula $D$.
\oomit{
MONA uses multi-terminal BDDs for efficiently representing the automata and performing operations such as product, projection, determinization and minimization, thereby permitting automaton construction for fairly large formulas, in spite of poor worst case complexity. 
}
We omit the details here. The reader may refer to several papers  on \qddc\/ 
for detailed description and examples of \qddc\/ specifications as well as  its model checking tool DCVALID \cite{Pan01a,MPW17,Pan01b,CP03,SPC05}.

\section{Supervisor and Controller}
\label{sec:definitions}
In this section we present \qddc\/  formulas and automata where variables $PV=I \cup O$ are partitioned into disjoint sets of input variables $I$ and output variables $O$. It is known that supervisors and controllers can be expressed as Mealy machines with special properties. Here we  show how Mealy machines can be represented as special form of Deterministic finite automata (DFA).  This representation allows us to use the MONA DFA library \cite{Mon01} to  compute supervisors and controllers efficiently using our tool \DCSYNTH.
\begin{definition}[Output-nondeterministic Mealy Machine]
\label{def:nondetmm}
A total and Deterministic Finite Automaton (DFA) over input-output alphabet $\Sigma=2^I \times 2^O$ is a tuple $A=(Q,\Sigma,s,\delta,F)$ 
having conventional meaning, where $\delta:Q \times 2^I \times 2^O \rightarrow Q$. An {\bf output-nondeterministic Mealy machine} is a DFA with a unique reject (or non-final) state $r$ 
which is a sink state i.e., $F= Q - \{r\}$ and $\delta(r,i,o)=r$ for all $i \in 2^I$, $o \in 2^O$.  \qed
\end{definition}

The intuition behind this definition is that the transitions from $q \in F$ to $r$ are
forbidden (and kept only for making the DFA total).
The language of any such Mealy machine is prefix-closed. 
Recall that for a Mealy machine, $F=Q-\{r\}$.
A Mealy machine is 
{\bf deterministic} if $\forall s \in F$, $\forall i \in 2^I$, $\exists$ at most one $o \in 2^O$  such that $\delta(s,i,o) \not=r$. 

\begin{definition}[Non-blocking Mealy Machine]
An output-nondeterministic Mealy machine is called {\bf non-blocking} if $\forall s \in F$, $\forall i \in 2^I$ $\exists o \in 2^O$ such that
$\delta(s,i,o) \in F$. It follows that 
for all input sequences a non-blocking Mealy machine can produce one or more output sequence without ever getting into the reject state. \qed 
\end{definition}
For a Mealy machine $M$ over variables $(I,O)$, its language $L(M) \subseteq (2^I \times 2^O)^*$. A word $\sigma \in L(M)$ can also be represented as pair $(ii,oo) \in ((2^I)^*,(2^O)^*) $ such that
$\sigma[k] = ii[k] \cup oo[k], \forall k \in dom(\sigma)$. Here $\sigma, ii, oo$ must have the same length. Note that in the rest of this paper, we do not distinguish between $\sigma$ and $(ii,oo)$.
Also, for any input sequence $ii \in (2^I)^*$, we define $M[ii] = \{ oo ~\mid~ (ii,oo) \in L(M) \}$. 

\begin{definition}[Controllers and Supervisors]
 An output-nondeterministic Mealy machine which is non-blocking is called a {\bf supervisor}.  An deterministic supervisor is called a {\bf controller}. \qed
\end{definition}
The non-deterministic choice of outputs in a supervisor denotes unresolved decision.
The determinism ordering defined below allows supervisors to be refined into controllers.

\oomit{\color{blue}
See Figure \ref{fig:mpnc} in Appendix \ref{section:2cellexample} for an example of a output-nondeterministic Mealy machine which is non-blocking.
We represent output-nondeterministic Mealy machines and controllers as DFAs to take advantage of well established 
BDD based  DFA libraries such as those in tool MONA \cite{Mon01}. This library is used to represent and manipulate the supervisors in tool \DCSYNTH.
}

\begin{definition}[Determinism Order and Sub-supervisor]
 Given two supervisors $S_1$ and  $S_2$, we say  $S_1 \leq_{det} S_2$ ($S_2$ is {\em more deterministic} than $S_1$),
 iff $L(S_2) \subseteq L(S_1)$.  We call $S_2$ to be a {\em sub-supervisor} of 
 $S_1$. \qed
 \end{definition}
Note that being supervisors, they are both non-blocking and hence 
$\emptyset \subset S_2[ii] \subseteq S_1[ii]$ for any $ii \in (2^I)^*$.
The supervisor $S_2$ may make use of additional memory for resolving and pruning
the non-determinism in $S_1$. 

For technical convenience, we define a notion of \textit{indicator variable} for a \qddc\/ formula
(regular property). The idea behind this is that the indicator variable $w$ witnesses the truth of a formula $D$ at any point in execution. Thus,
\vspace{-1ex}
\[
   Ind(D,w) \df pref(EP(w) \Leftrightarrow D)
\]
Here, $\mathbf{EP(w)} = (true \verb#^# \langle w \rangle)$, i.e. $EP(w)$ holds at a point $i$, if variable $w$ is true at that point $i$. Hence, $w$ will be $true$ exactly on those points where $D$ is $true$. The
formula automaton $\A(Ind(D,w))$ gives us a controller with input-output alphabet $(I\cup O,w)$ such that it outputs $w=1$ on a transition iff the past satisfies $D$. Since our formula automata are minimal DFA, 
$\A(Ind(D,w))$ characterizes the least memory needed to track the truth
of formula $D$.

\section{\DCSYNTH\/ Specification and  Controller Synthesis}
\label{section:dcsynth-spec}
This section defines the \DCSYNTH\/ specification and presents the algorithm used in our tool \DCSYNTH for soft requirement guided controller synthesis from a \DCSYNTH specification. The process of synthesizing a controller as discussed in Section \ref{sec:synthesisAlgo} uses three main algorithms
given in Sections \ref{sec:MPNCConstruction}-\ref{sec:determinization}.
\vspace{-1ex}
\subsection{Invariance Properties and Maximally Permissive Supervisor}
\label{sec:MPNCConstruction}

A \qddc\/ formula $D$ specifies a regular property which may hold intermittently during a behaviour (see Definition \ref{def:dclang}).
An important class of properties, denoted by $\invariant ~D$, states that $D$ must hold invariantly during the system behaviour.

\begin{definition}
\label{def:realizableAndMPNC}
Let ${\mathbf{S}}~ \mathbf{realizes} \mathbf{~\invariant~ D}$ denote that a supervisor $S$  realizes {\em invariance} of  \qddc formula $D$ over variables $(I,O)$. Define ${\mathbf{S}}~ \mathbf{realizes} \mathbf{~\invariant~ D}$ provided $L(S) \subseteq L(D)$. Recall that, by the definition of supervisors, $S$ must be non-blocking.
A supervisor $S$ for a formula $D$ is called {\bf maximally permissive} $\mathit{iff}$ $S \leq_{det} S'$ holds for any supervisor $S'$ such that  $S' ~\mbox{\textbf{realizes}} ~\invariant~ D$. This $S$ (when it exists) is unique upto language equivalence of automata, and the minimum state maximally permissive supervisor is denoted as $\textbf{\MPNC}(D)$. \qed
 \end{definition}


Now, we discuss how $\MPNC(D)$  for a given \qddc formula $D$ is computed.
\vspace{-1ex}
\begin{enumerate}
\item Language equivalent DFA $\A(D)=\langle S,2^{I\union O},\delta,F\rangle$
is constructed for formula $D$ (Theorem \ref{theorem:formula-automaton}). 
The standard safety synthesis algorithm \cite{GTW02} over  
$\A(D)$ gives us the desired $\MPNC(D)$ as outlined in the following steps.

\item We first compute the \emph{largest} set of winning states $G \subseteq F$ with the following property: $s \in G$ iff $\forall i\exists o:\delta(s,(i,o)) \in G$. Let 
$Cpre(\A(D),X)= \{s ~\mid~ \forall i \exists o:\delta(s,(i,o)) \in X \}$. 
Then we iteratively compute $G$ as follows:\\ 
\hspace*{0.4cm} G=F; \\ 
\hspace*{0.5cm}{\textbf do} \\ 
\hspace*{1cm}G1=G; \\
\hspace*{1cm}G=Cpre($\A(D)$,G1); \\
\hspace*{0.5cm}{\textbf while} (G != G1); 

\item If initial state $s \notin G$, then the specification is \emph{unrealizable}. Otherwise, $\MPNC(D)$ is obtained by declaring $G$ as
the set of final states and retaining all the transitions in $\A(D)$ between states in $G$ and redirecting the remaining
transitions of $\A(D)$ to a unique reject state $r$ which is made a sink state. 
\end{enumerate}
\vspace{-1ex}
\begin{proposition}
\label{prop:mpnccorrect}
For a given \qddc formula $D$ the above algorithm computes the maximally permissive supervisor $\MPNC(D)$. \qed
\end{proposition}
The proposition follows straightforwardly by combining Theorem \ref{theorem:formula-automaton} with the correctness of standard safety synthesis algorithm \cite{GTW02}. We omit a detailed proof.
\oomit{
The objective is to synthesize a 
deterministic controller $Cnt$ which (a) {\bf invariantly} satisfies the
hard requirement $D^h$, and (b) it is {\bf $H$-optimal for $D^s$} amongst all controllers satisfying (a).
$H$-optimality maximizes the Expected count of (intermittent) occurrence of soft requirement $D^s$ over next $H$ steps of execution. This count is averaged over all input sequences of length $H$. 
Appendix \ref{section:supervisor} gives a formal definition of controller/supervisor and a supervisor
invariantly satisfying a formula.
}

\vspace{-1ex}
\subsection{Maximally Permissive $H$-Optimal Supervisor (\GODSC)}
\label{sec:HoptimalComputation} 
Given a supervisor $S$ and a desired \qddc formula $D$ which should hold ``as much as possible'' (both are over input-output variables $(I,O)$), we give a method for constructing an ``optimal'' sub-supervisor of $S$, which maximizes 
the expected value of count of $D$ holding in next H moves when averaged over all the inputs.

First consider $\A^{Arena} = S \times \A(Ind(D,w))$ which is a supervisor over input-output variables $(I,O \cup \{w\})$. It augments
$S$ by producing an additional output $w$ which witnesses the truth of $D$. It has the property: $L(\A^{Arena}) \downarrow (I \cup O)=L(S)$. Also for $\sigma \in L(\A^{Arena})$ and $i \in dom(\sigma)$ we have
$w \in \sigma[i]$ iff $\sigma[0:i] \models D$. Thus, every transition of $\A^{Arena}$ is labelled with $w$ iff $D$ holds on taking the transition.
Let the weight of transitions labelled with $w$ be 1 and 0 otherwise. Thus, for $o \in 2^{(O \cup \{w\})}$ let $wt(o)=1$ if $w\in o$ and $0$ otherwise.
Technically, this makes $\A^{Arena}$ a weighted automaton.

In the supervisor $\A^{Arena}=(Q,\Sigma,s,\delta,Q-\{r\})$, where $r$ is the unique reject state, we define for $(q \in Q) \not= r$ and $i \in 2^I$, set $LegalOutputs(q,i)=\{ o ~\mid~ \delta(q,i) \not=r \}$.
We also define a deterministic selection rule as function $f$ s.t. ~$f(q,i) \in LegalOutputs(q,i)$ and a non-deterministic selection rule $F$ as function $F$ s.t.$F(q,i) \in \{ O \subseteq LegalOutputs(q,i) ~\mid~ O \neq \emptyset\}$.
Let $H$ be a natural number. Then $H$-horizon policy $\pi$ is a sequence $F_1,F_2, \ldots, F_H$ of non-deterministic selection rules. A deterministic policy will use only deterministic selection rules. A policy is stationary (memory-less) if each $F_i$ is the same independently of $i$.

Given a state $s$, a  policy $\pi$ and an input sequence   $ii \in (2^I)^H$ (of length $H$), we  define $L(\A^{Arena},ii,s)$ as all runs of $\A^{Arena}$ over
the input $ii$ starting from state $s$ and  $L^\pi(\A^{Arena},ii,s)$ as all runs over input $ii$ starting  from $s$ and following the selection rule $F_i$ at step $i$. Each run has the form $(ii,oo)$.  Let $Value(ii,oo) = \Sigma_{1 \leq i \leq \#ii} ~~wt(oo[i]))$.
Thus, $Value(ii,oo)$ gives the count of $D$ holding during behaviour fragment $(ii,oo)$.  Then, we define
$VMIN^\pi(s,ii) = min \{ Value(ii,oo) ~\mid~ (ii,oo) \in L^\pi(\A^{Arena},ii,s)\}$, which 
gives the minimum possible count of $D$  among all the runs of $S$ under policy $\pi$ on input $ii$,  starting with state $s$.
We also define,  $VMAX(s,ii) = max \{ Value(ii,oo) ~\mid~ (ii,oo) \in L(\A^{Arena},ii,s)\}$, which  gives the maximum achievable count. 
Note that $VMAX$ is independent of any policy.

Given a horizon value (natural number) $H$, $\A^{Arena}$ and a non-deterministic $H$-horizon policy $\pi$, we define utility values $ValAvgMin^\pi(s)$ and $ValAvgMax(s)$ for each state $s$ of $\A^{Arena}$ as follows.
\vspace{-1ex}
\[
\begin{array}{l}
 ValAvgMax(s) ~=~ \mathbb{E}_{ii \in (2^I)^H} ~VMAX(s,ii) \\
 ValAvgMin^\pi(s) ~=~ \mathbb{E}_{ii \in (2^I)^H} ~VMIN^\pi(s,ii) 
\end{array}
\]
Thus, intuitively, $ValAvgMax(s)$ gives the maximal achievable count of $D$ from state $s$, when averaged over all inputs of length $H$.
Similarly, $ValAvgMin^\pi(s)$ gives the minimal such count for $D$ under policy $\pi$, when averaged over all inputs of length $H$.
Our aim is to construct a horizon-$H$ policy $\pi^* = argmax_\pi ~~ValAvgMin^{\pi}(s)$. This will turn out to be a stationary policy given by a selection rule $F^*$. This rule can be implemented as a supervisor denoted by $\GODSC(\A^{Arena},H)$. We now give its construction.

The well known value iteration algorithm allows us to efficiently compute $ValAvgMax(s)$ as recursive function $Val(s,H)$ below.
\vspace{-1ex}
\[
 \begin{array}{l}
 Val(s,0) = 0 \\
 Val(s, p+1) = \mathbb{E}_{i \in 2^I}~~ max_{o \in 2^{(O \cup \{w\})} ~:~ \delta(s,(i,o)) \neq r}~  \\
 \hspace*{4cm} \{ wt(o) ~+~ Val(\delta(s,(i,o)), p) \} 
\end{array}
\]
We omit the straightforward proof that $Val(s,H)=ValAvgMax(s)$ (see \cite{Put94}).

Having computed this, the optimal selection rule $F^*$ giving stationary policy $\pi^*$ is given as follows: 
For each state $s \in \A^{Arena}$ and each input $i\in 2^I$, 
\vspace{-1ex}
\[
 \begin{array}{l}
 F^*(s,i) =  argmax_{o \in 2^O} \{ wt(o) ~+~  Val(s,H) \\
 \hspace*{2cm} ~\mid~  \delta_{\A^{Arena}}(s,(i,o)) = s' \land s'\not=r \} 
\end{array}
\]
Note that $F^*(s,i)$ is non-deterministic as more than one output $o$ may satisfy the $argmax$ condition. The following well-known lemma
states that stationary policy $\pi^*$ using the selection rule $F^*$ is $H$-optimal.
\begin{lemma}
\label{lemma:godsccorrect1}
 For all states $s$ of $\A^{Arena}$, $ValAvgMin^{\pi^*}(s) =
 ValAvgMax(s)$ always holds. Therefore, for all states $s$ of $\A^{Arena}$ and for any $H$-horizon policy $\pi$, $ValAvgMin^\pi_H(s) \leq ValAvgMin^{\pi^*}(s)$ also holds. \qed
\end{lemma}
We omit the proof of these well known properties from optimal control of Markov Decision Processes (see \cite{Put94}).

Supervisor $\A^{Arena}$ is pruned to retain only the transitions with the  outputs in set $F^*(s,i)$ (as these are all equally optimal). 
This gives us  {\em Maximally permissive $H$-Optimal sub-supervisor of $\A^{Arena}$ w.r.t. $D$}.
This supervisor is denoted by $\GODSC(\A^{Arena},H)$ or 
equivalently $\GODSC(S,D,H)$.
The following proposition follows 
immediately from the construction of $\GODSC(S,D,H)$ and Lemma \ref{lemma:godsccorrect1}.  
\begin{proposition}
\label{prop:godsccorrect2}
\begin{enumerate}
 \item $S \leq_{det} \GODSC(S,D,H)$, for all $H$.
 \item $\GODSC(S,D,H)$ is maximally permissive $H$-optimal sub-supervisor of $S$.
 \item If $\GODSC(S,D,H) \leq_{det} S'$ then $S'$ is $H$-Optimal. \qed 
\end{enumerate}
\end{proposition}

\oomit{
The following lemma states some of its properties. Note that every sub-supervisor $S'$ of $S$ defines a selection
rule and hence a stationary policy. Hence we can, with mild abuse of notation, use  the term $ValAvgMin^{S_2}(s)$.
\begin{lemma}
Let $GD=\GODSC(\A^{Arena},H)$ and let $S_1$ be any sub-supervisor of $\A^{Arena}$ and $S_2$ be any sub-supervisor of $GD$. Then,
\begin{itemize}
 \item For all states $s$ of $\A^{Arena}$, we have, $ValAvgMin^{S_1}(s) \leq ValAvgMin^{GD}(s)$.
 \item For all states $s$ of $\A^{Arena}$, we have $ValAvgMin^{S_2}(s) = ValAvgMin^{GD}(s)$.
 \item If for all states $s$ of $\A^{Arena}$, we have  $ValAvgMin^{S_1}(s) \geq ValAvgMin^{GD}_H(s)$, then
 $S_1$ is a sub-supervisor of $GD$. (Hence, $GD$ is maximally permissive.)
\end{itemize}
\end{lemma}
}

\subsection{From Supervisor to Controller}
\label{sec:determinization}
A controller $Cnt$ can be obtained from a supervisor $S$ by resolving output non-determinism in $S$. We give a rather straightforward
mechanism for this. We allow the user to specify an ordering $Ord$ on the set of output variables $2^O$. A given supervisor $S$ is determinized by retaining only the highest ordered output among those permitted by $S$. This is denoted  $Det_{Ord}(S)$. The output ordering is specified by giving a lexicographically ordered list of output variable literals. This facility is used to determinize $\GODSC$ and $\MPNC$ supervisors as required.
\vspace{-2ex}
\begin{example}
For a supervisor $S$ over variables $(I,\{o_1,o_2\})$, an example output order can be given as lexicographically ordered list ($o_1 > ~!o_2$).  Then, for any transition the determinization step will try to select the highest ordered output (which is allowed by $S$) from the list \{($o_1=true, o_2=false$), ($o_1=true, o_2=true$), ($o_1=false, o_2=false$), ($o_1=false, o_2=true$)\}. \qed
\end{example}

\subsection{\DCSYNTH Specification and Controller Synthesis}
\label{sec:synthesisAlgo}
A {\bf \DCSYNTH specification} is a tuple 
$(I,O,D^h, D^s)$, where $I$ and $O$ are the set of {\em input} and {\em output} variables, respectively. Formula $D^h$ called the \emph{hard requirement} and formula $D^s$ called the \emph{soft requirement} are \qddc\ formulas over the set of propositions $PV=I \cup O$. Let $H$ be a natural number called Horizon.
The  objective in \DCSYNTH\/ is to synthesize a deterministic controller which 
(a) {\em invariantly} satisfies the hard requirement $D^h$, and 
(b)  it is $H$ Optimal w.r.t. $D^s$ amongst all the controllers satisfying (a).

Given a specification $(I,O,D^h,D^s)$, a horizon value $H$ (a natural number) and a total ordering $Ord$ on the set of outputs $2^O$, the controller synthesis in \DCSYNTH\/ can be given as Algorithm \ref{algo:synthesis}.
\oomit{
gives the outline of synthesis method. We first computes supervisor $\MPNC(D^h)$ denoted by $\A^{mpnc}$ as given in section \ref{sec:MPNCConstruction}. The \DCSYNTH specification $(I,O,D^h, D^s)$  is said to be realizable iff $\MPNC(D^h)$ is realizable (i.e. it exist with start state as accepting state). If specification is not realizable then algorithm terminates. Otherwise, 
in the next step a sub-supervisor of $\MPNC(D^h)$ which satisfies $D^s$ for ``as many inputs as possible'' is computed by using notion of $H$-optimality w.r.t. the soft requirement $D^s$ as given in section \ref{sec:HoptimalComputation} to get $A^{godsc}$. We then determinize the $A^{godsc}$ using default output ordering as given in section \ref{sec:determinization} to get a controller $Cnt$.
}
\vspace{-1ex}
\begin{algorithm}
\label{algo:synthesis}
\textbf{ControllerSynthesis} \\
\textbf{Input}: $S=( I,O,D^h,D^s)$. Horizon $H$, Output ordering $Ord$ \\
\textbf{Output}: Controller Cnt for S. \\[1ex]
1. $\A^{\mpnc}$ = $\MPNC(D^h)$ \\
2. $\A^{\godsc}$ = $\GODSC(\A^{\mpnc}, D^s, H)$ \\
3. $Cnt$ = $Det_{ord}(\A^{\godsc})$. \\
4. Encode the automaton $Cnt$ in an implementation language.
\end{algorithm}
Step 1 uses the \MPNC construction given in Section \ref{sec:MPNCConstruction}.
Step 2 uses the \GODSC construction given in Section \ref{sec:HoptimalComputation}  whereas
Step 3 uses the determinization method of Section \ref{sec:determinization}.

\vspace{-1ex}
\begin{proposition}
The controller $Cnt$ output by Algorithm  \ref{algo:synthesis} invariantly satisfies $D^h$, and it intermittently, but $H$-optimally, satisfies $D^s$.
\vspace{-1ex}
\end{proposition}
\begin{proof}
By Proposition  \ref{prop:mpnccorrect}, $\A^{\mpnc}$ realizes $\invariant~ D^h$.
Then, by Proposition \ref{prop:godsccorrect2}, $\A^{\godsc}$ and $Cnt$
are sub-supervisors of  $\A^{\mpnc}$ and hence they also realize $\invariant ~D^h$. Moreover, by Lemma \ref{lemma:godsccorrect1}, we get that
$\A^{\godsc}$ is $H$-optimal w.r.t. $D^s$. Hence, by Proposition \ref{prop:godsccorrect2}, we  get that $Cnt$ which is a sub-supervisor of  $\A^{\godsc}$ is also $H$-Optimal with respect to $D^s$. \qed
\vspace{-1ex}
\end{proof}
At all stages of above synthesis, the automata/supervisors $\A(D^h)$, $\A(D^s)$, $A^{\mpnc}$ and $A^{\godsc}$ and $Cnt$ 
are all represented as semi-symbolic automata (SSDFA) using the MONA \cite{Mon01} DFA data structure. In this representation, the transition function is represented as a multi-terminal BDD. MONA DFA library provides a rich set of automata operations including product, projection, determinization and minimization over the SSDFA. 
The algorithms discussed in Sections \ref{sec:MPNCConstruction}, \ref{sec:HoptimalComputation} and \ref{sec:determinization} are implemented over SSDFA.
Moreover, these algorithms are adapted to work without actually expanding the specification automata into game graph. At each stage of computation, the automata and supervisors are aggressively minimized, which leads to significant improvement in the scalability and computation time of the tool.  Appendix \ref{sec:ssdfa} gives the details of SSDFA data structure and its use in symbolic computation of supervisors in efficient manner.

\vspace{-1ex}
\section{Case Studies and Experiments}
\label{section:motivation}
\vspace{-1ex}
For a \DCSYNTH specification, $D^h$ and $D^s$ can be any \qddc formulas. While invariance of $D^h$ is guaranteed by the synthesis algorithm, the quality of the controller is controlled by optimizing the outputs for which the soft requirement $D^s$ holds. For example, $D^s$ may specify outputs which save energy, giving an energy efficient controller.
The soft requirement can also be used to improve the robustness \cite{BCGHHJKK14} of the controller (see \cite{WPM17}). Below, we consider specifications structured as assumptions and commitments, and their optimized robustness using our soft requirement guided synthesis.
\vspace{-1ex}
\subsection{Types of Controller Specification}
\label{sec:controllertype}

For many examples, the controller specification can be given as a pair $(A,C)$ of
\qddc\/ formulas over input-output variables $(I,O)$. Here, {\bf commitment} $C$
is a formula specifying the desired behaviour which must ideally hold invariantly.
But this may be unrealizable, and a suitable {\bf assumption} $A$ on the behaviour of environment may have to be made for $C$ to hold. In case the assumption $A$ does not hold,
it is still desirable that controller satisfies $C$,  intermittently but ``as much as
possible''.  Given this assumption-commitment pair $(A,C)$,  we specify four types of
derived controller specifications $(I,O,D^h,D^s)$ as follows.
\vspace*{-0.2cm}
\begin{center}
{\scriptsize
\begin{tabular}{|c|c|c|}
\hline
\textbf{Type} & \textbf{Hard Requirement} $D^h$ & \textbf{Soft Requirement} $D^s$ \\
\hline
\textit{Type0} & $C$ & $true$ \\
\hline
\textit{Type1} & $(A \Rightarrow C)$ & $true$ \\ 
\hline
\textit{Type2} & $true$ & $C$ \\
\hline
\textit{Type3} & $(A \Rightarrow C)$ & C \\
\hline
\end{tabular}
}
\end{center}
\textit{Type0} controller gives the best guarantee but it may be unrealizable. \textit{Type1} controller
provides a firm but conditional guarantee. \textit{Type2} controller tries to achieve 
$C$ in $H$-optimal fashion irrespective of any assumption and \textit{Type3} Controller provides firm conditional guarantee and it also tries to satisfy $C$ in $H$-optimal fashion even when the assumption does not hold.

\vspace*{-0.3cm}
\subsection{Performance Metrics: Measuring quality of controllers}
\label{sec:performanceMetrics}
For the same assumption commitment pair $(A,C)$, we can synthesize diverse controllers
using different specification types, horizon values and output orderings.
In order to compare the performance of these different controllers,  
we define two metrics -- i) \textit{Expected Case Performance measure} to compare average case behaviour, 
and ii) \textit{Must Dominance} to compare the guaranteed behaviour. 
\vspace*{-0.2cm}
\paragraph{i) Expected Case Performance:}
Given a controller $Cnt$ over input-output alphabet $(I,O)$ and a \qddc\/ formula (regular property) $C$ over variables $I \cup O$, we can construct a {\em Discrete Time Markov Chain (DTMC)}, denoted $M_{unif}(Cnt,C)$, whose analysis allows us to measure the probability of $C$ holding in long runs (steady state) of $Cnt$ under random  independent and identically distributed(iid) inputs. This value is designated as $\mathbb{E}_{unif}(Cnt,C)$. The construction of the desired DTMC 
is as follows. The product $Cnt \times \A(C)$ gives a finite state automaton with the same behaviours as $Cnt$. Moreover, it is in accepting state exactly when $C$ holds for the past behaviour. (Here $\A(C)$ works as a total deterministic monitor automaton for $C$ without restricting $Cnt$). By assigning uniform discrete probabilities to all the inputs from any state, we obtain the DTMC $M_{unif}(Cnt,C)$ along with a designated set of accepting states, such that the DTMC is in accepting state precisely when $C$ holds.
Standard techniques from Markov chain analysis allow us to compute the {\em probability} (Expected value) of being in the set of accepting states on long runs (steady state) of the DTMC. This gives us the desired value $\mathbb{E}_{unif}(Cnt,C)$.
A leading probabilistic model checking tool MRMC implements this computation \cite{KZHHJ11}.
In \DCSYNTH, we provide a facility to compute $M_{unif}(Cnt,C)$ in a format accepted by the tool MRMC. Hence, using \DCSYNTH\/ and  MRMC, we are able to compute  $\mathbb{E}_{unif}(Cnt,C)$.
\vspace*{-1ex}
\paragraph{ii) Guaranteed Performance as Must-Dominance:}
Consider two supervisors $S_1$, $S_2$ and a regular property $C$. Define that $S_i$ guarantees $C$ for an input sequence $ii$, provided for every output
sequence $oo \in S_i[ii]$ produced by $S_i$ on $ii$ we have  that $(ii,oo)$ satisfies $C$. 
We say that $S_2$  \emph{must dominance} $S_1$ with respect to the property $C$ provided for every input sequence $ii$, if $S_1$ guarantees $C$ then
$S_2$ also guarantees $C$. Thus, $S_2$ provides a superior must guarantee of $C$ than $S_1$.
%
\begin{definition}[Must Dominance] Given two supervisors $S_1, S_2$ and a property (formula) 
$C$ over input-output alphabet $(I,O)$, 
the must dominance of $S_2$ over $S_1$ is defined as 
$S_1 \leq_{dom}^C S_2  \quad \mbox{iff} \quad$ 
$MustInp(S_1,C) ~\subseteq~ MustInp(S_2,C)
$, 
where $MustInp(S_i,C) = \{ ii \in (2^I)^+ \mid \forall oo \in (2^O)^+. ( (ii,oo) \in L(S_i) \Rightarrow (ii,oo) \models C \}$.  \qed
\vspace{-1ex}
\end{definition}
\oomit{
To compare the behaviours of any two given supervisors guaranteed in-terms of its input sequences that satisfy a given regular property, the \emph{must dominance} measure can be used. 
}
We establish \emph{must dominance} relations among \GODSC supervisors of various types of specifications discussed in Section \ref{sec:controllertype}. 
\vspace{-1ex}
\begin{lemma}
\label{lemma:mustdom}
For any \qddc formulas $A$ and $C$, and any horizon $H$, the following must dominance relations will hold (for any given $H$)
\vspace{-1ex}
\begin{enumerate}
{\scriptsize
    \item $\GODSC_1(A,C))$ $ \leq_{dom}^C \GODSC_3(A,C)) \leq_{dom}^C \GODSC_0(A,C))$ 
    \item $\GODSC_2(A,C)) \leq_{dom}^C \GODSC_0(A,C))$
}
\end{enumerate}
where, $\GODSC_i(A,C)$ denote the maximally permissive $H$-optimal supervisor $\A^{\GODSC}$ of Algorithm \ref{algo:synthesis} for the specification $Type_i(A,C)$. 
\vspace{-1ex}
\end{lemma} 
\begin{proof}
By definition, {\scriptsize $\GODSC_0(A,C)$} invariantly satisfies $C$ for all input sequences. Hence,
{\scriptsize $MustInp(\GODSC_0(A,C),C)= (2^I)^*$}, which immediately gives us that {\scriptsize $S \leq_{dom}^C \GODSC_0(A,C))$} for any supervisor $S$.

Now we prove the remaining relation {\scriptsize $\GODSC_1(A,C))$ $ \leq_{dom}^C \GODSC_3(A,C))$}.
Let {\scriptsize $S=\MPNC(A \Rightarrow C)$}. Then, {\scriptsize $\GODSC_1(A,C))=\GODSC(S,true,H)=S$}. The
second equality holds as soft requirement $true$ does not cause any pruning of outputs
in $H$-optimal computation. By definition
{\scriptsize $\GODSC_3(A,C)=\GODSC(S,C,H)$}. By Proposition \ref{prop:godsccorrect2}, {\scriptsize $S \leq_{det} \GODSC(S,C,H)$}
which gives us the required result. \qed
\end{proof}
Note that in general, {\scriptsize $\GODSC_2(A,C)$} is theoretically incomparable with {\scriptsize $\GODSC_1(A,C)$} and  {\scriptsize $\GODSC_3(A,C)$}, as  {\scriptsize $\GODSC_2(A,C)$} is a supervisor that does not have to meet any hard requirement, but it optimally meets the soft requirements irrespective of the assumption. However, for specific $(A,C)$ instances, some additional must-dominance relations may hold between {\scriptsize $\GODSC_2(A,C)$} and the other supervisors.

\vspace*{-0.3cm}
\subsection{Case Studies: Mine-pump and Arbiter Specifications}
\label{section:minepumpcasestudy}
We have carried out experiments with  i) the Mine-pump specification presented in this section, and ii) an Arbiter specification
given in  Appendix \ref{sec:casestudy_arbiter}.

\noindent \textbf{Mine-pump:}
The Mine-pump controller (see \cite{Pan01a}) has two input sensors: high water level sensor {\scriptsize $\mathit{HH2O}$} and methane leakage sensor {\scriptsize $\mathit{HCH4}$}; 
and one output, 
{\scriptsize $\mathit{PUMPON}$} to keep the pump on. 
The objective of the controller is to \emph{safely} operate 
the pump in such a way that  the water level never remains high continuously for more
that $w$ cycles. 
Thus, Mine-pump controller specification has input and output variables  {\scriptsize $(\{HH2O, HCH4\}, \{PUMPON\})$}.

We have following {\bf assumptions} on the \textit{mine} and the \textit{pump}. Their conjunction
is denote {\bf $MineAssume(\epsilon,\zeta,\kappa)$} with integer parameters $\epsilon,\zeta,\kappa$. Being of the form $[]D$ each formula states that the property $D$ (described in text) holds for all observation intervals in past.
\vspace{-1ex}
\begin{itemize}
\item[-] \emph{Pump capacity:} {\scriptsize $([]!(slen=\epsilon ~\&\& ~([[PUMPON ~\&\& ~HH2O]]$\verb|^|$\langle HH2O \rangle)))$}.
If the pump is continuously on for $\mathit{\epsilon}$ cycles with water level also continuously high, 
then water level will not be high at the $\mathit{\epsilon+1}$ cycle. 
\item[-] \emph{Methane release:} 
{\scriptsize $[](([HCH4]$\verb|^|$[!HCH4]$\verb|^|$\langle HCH4 \rangle ) \Rightarrow (slen> \zeta))$} and
{\scriptsize $[]([[HCH4]] \Rightarrow slen<\kappa)$}.
The minimum separation between the two leaks of methane is {\scriptsize $\mathit{\zeta}$} cycles 
and the methane leak cannot persist for more than $\mathit{\kappa}$ cycles. 
\vspace{-1ex}
\end{itemize}
\noindent The {\bf commitments} are as follows. The conjunction of commitments is denoted by {\bf $MineCommit(w)$} and they hold intermittently in absence of assumption.
\vspace{-1ex}
\begin{itemize}
\item[-] \emph{Safety conditions:} 
{\scriptsize $true$\verb|^|$\langle((HCH4 ~|| ~!HH2O) \Rightarrow ~!PUMPON))\rangle$} states that
if there is a methane leak or absence of high water in current cycle, then pump should be off in the current cycle. Formula
{\scriptsize $!(true$\verb|^|$([[HH2O]] ~\&\& ~slen = w))$} states that the water level does not  remain continuously high in last $w+1$ cycles.
\vspace{-1ex}
\end{itemize}
The Mine-Pump specification denoted by {\bf $MinePump(w,\epsilon,\zeta,\kappa)$} is given by the assumption-commitment pair
$(MineAssume(\epsilon,\zeta,\kappa), MineCommit(w))$. The four types of \DCSYNTH\/ specifications of Section \ref{sec:controllertype}
can be derived from this.  Figure  \ref{figure:minepumpspecification} in Appendix \ref{sec:exampleSources} gives the textual source of $Type3(MinePump(8,2,6,2))$
specification used by the \DCSYNTH\/ tool. 

\noindent \textbf{Arbiter:}
Due to space limitations, the detailed specification of the arbiter, briefly discussed as Example \ref{exam:arbiter} in Section \ref{section:intro}, is 
given in  Appendix \ref{sec:casestudy_arbiter}. The arbiter is denoted as $Arb(n,k,r)$, where $n$ denotes the number of clients, 
$k$ is the response time (time for which a client should keep the  request high continuously to get the guaranteed access) and $r$ is the maximum number of request that can be true simultaneously. 

\subsection{Experimental Evaluation}\label{sec:experiments}
\vspace{-1ex}
Given an assumption-commitment pair $(A,C)$ the four types of \DCSYNTH\/ specifications
can be derived as given in Section \ref{sec:controllertype}. Given any such specification, a horizon value $H$, and an
ordering of outputs, a controller can be synthesized using our tool \DCSYNTH as described in Section \ref{sec:synthesisAlgo}.
\oomit{ 
We synthesized the controllers for all the four types of specifications given in Section \ref{sec:controllertype} using \DCSYNTH for both mine pump and the arbiter considering i) $MinePump(8,2,6,2)$  and ii) $Arb(5,3,2)$. 
}
For the \textit{Mine-pump} instance $MinePump(8,2,6,2)$, we synthesized controllers for all the four derived specification types with horizon value $H=50$ and output ordering $PUMPON$. These controllers choose to get rid of water aggressively
by keeping the pump on whenever possible. Similarly, controllers were also synthesized with the output ordering $!PUMPON$. These controllers save energy by keeping the pump off whenever possible. Note that, in our synthesis method, hard and soft requirements are fulfilled before applying the output orderings.

For the Arbiter instance $Arb(5,3,2)$ also, controllers were synthesized
for all the four derived specification types with horizon value $H=50$ and output ordering  $ArbDef=(a_1>a_2>a_3>a_4>a_5)$. This ordering tries to give acknowledgment such that  client $i$ has priority higher than client $j$ for all $i < j$.

\begin {table}[t]
\caption {Synthesis from Mine-pump(8,2,6,2) and Arb(5,3,2) specifications in \DCSYNTH. The last column gives the expected value of commitment in long run on random inputs.} 
\label{tab:expectedValueMeasurementNew}
\vspace{-3ex}
\begin{center}
{\scriptsize
	\begin{tabular}	{|c|c|c||c|c|c||c|}
	\hline
	& 
	\multicolumn{2}{c||}{ \textbf{\DCSYNTH Specification}} 
	& \multicolumn{3}{c||}	 {\textbf{ Synthesis (States/Time)}}
	& 
	\\
	\hline
		\textbf{Sr} &
		\textbf{Controller}  & 
		\textbf{Output}  & 
		\textbf{\MPNC} & 
		\textbf{\GODSC} &
	 	\textbf{Controller} & 
		\textbf{Expected}  
		\\
		\textbf{No} & \textbf{type}
		  &
		\textbf{Ordering} &
		\textbf{Stats} &
		\textbf{Stats} &
		\textbf{Stats} &
		\textbf{Value}
		\\
	\hline
	\multicolumn{7}{|c|}{$\mathbf{Mine-pump(8,2,6,2)}$} \\
     \hline	
		1 & 
	    $Type0$  & 
	      - &
	      \multicolumn{3}{c||}{Unrealizable} &
      
	     \\
	     \hline
	    2 & 
	    $Type1$  & 
	      $PUMPON$ & 
	      70/0.00045 &
	      70/0.00254 &
	      21/0.00220 & 
		  0.0
	     \\
	     \hline
		3 & 
		$Type2$  & 
	     $PUMPON$ & 
	      1/0.00004 &
	      10/0.00545 &
	      10/0.00033 &
		  0.99805
	     \\
	     \hline
	     4 & 
	    $Type3$  & 
	     $PUMPON$ & 
	     70/0.00045 &
	     75/0.044216 &
	     73/0.00081 &
	     0.99805
	     \\
	     \hline
	     \hline
	    5 & 
	    $Type1$  & 
	     $!(PUMPON)$ & 
	     70/0.00045 &
	      70/0.00254 &
	      47/0.00230 &
	      0.0
	     \\
	     \hline
		6 & 
	    $Type2$  & 
	     $!(PUMPON)$ & 
	      1/0.00004 &
	      10/0.00545 &
	      10/0.00019 &
		  0.99805
	     \\
	     \hline
	    7 & 
	    $Type3$  & 
	     $!(PUMPON)$ & 
	     70/0.00045 &
	     75/0.044216 &
	     73/0.00082 &
	      0.99805
	     \\
	     \hline
	     \hline
	\multicolumn{7}{|c|}{$\mathbf{Arb(5,3,2)}$} \\
		\hline	
		1 & 
	    $Type0$  & 
	     - & 
	     \multicolumn{3}{c||}{Unrealizable} &
	     \\
	     \hline
	    2 & 
	    $Type1$  & 
	     $ArbDef$  & 
	      13/0.000226 &
	      13/0.004794 &
	      11/0.007048 & 
		  0.0	
	     \\
	     \hline
		3 & 
	    $Type2$  & 
	     $ArbDef$ & 
	     1/0.00001 &
	     207/1.864346 &
	      201/0.058423 &
	      0.9930985
	     \\
	     \hline	
		4 & 
	    $Type3$  & 
	     $ArbDef$ & 
	       13/0.000213   &
	      207/1.897907 &
	      201/0.057062 & 
		  0.9930985  
	     \\
	     \hline	
	\end{tabular}
}
\vspace{-5ex}
\end{center}
\end{table}
In Table \ref{tab:expectedValueMeasurementNew} we give the performance of the of tool \DCSYNTH in synthesizing these controllers. The table gives the time taken at each stage of the synthesis algorithm, and the sizes of the computed supervisors/controllers. 
The experiments were conducted on Linux (Ubuntu 16.04) system with Intel i5
64 bit, 2.5 GHz processor and 4 GB memory.
\oomit{
$\mathit{PumpOff}$ i.e. $!\mathit{PumpOn}$. 
The first controller tries to keep the $PumpOn$ whenever possible and hence it gets rid of water promptly, whereas the second controller tries to keep pump off trying to save energy whenever possible.

For \textit{Arbiter}, the default value denoted by $ArbDef=(a_1>a_2>a_3>a_4>a_5)$ is used.  It states -- try to give acknowledgment to every client 
with client $i$ having priority higher than that of client $j$ for all $i < j$. 
}

\textit{Experimental Evaluation of Expected Case Performance:} 
The last column of Table \ref{tab:expectedValueMeasurementNew} gives the expected value of commitment holding in long run for the controllers of various types for both Mine-pump and Arbiter instances. This value is computed as outlined in Section \ref{sec:performanceMetrics}. The results are quite encouraging. 

It can be observed from Table \ref{tab:expectedValueMeasurementNew} that in both the examples, the controllers for \textit{Type1} (i.e., when soft-requirements are not used) specifications have $0$ expected value of commitment $C$. This is because of the strong assumptions used in guaranteeing $C$, which themselves have expected value $0$. In such a case, whenever the assumption fails, the synthesis algorithm has no incentive to try to meet $C$.

On the other hand, with soft requirement $C$ in \textit{Type2} and \textit{Type3} specifications, the  $H$-optimal controllers have the expected value of $C$ above $99\%$. This remarkable increase in the \textit{expected value} of Commitment shows 
that $H$-optimal synthesis is  very effective in figuring out controllers which meet the desirable property $C$ as much as possible, irrespective of the assumption.

\paragraph{Experimental Evaluation of Must-Dominance:} Given supervisors $S_1, S_2$ for
an assumption-commitment pair $(A,C)$, since both $S_1,S_2$ are finite state Mealy machines and $C$ is a regular property,
an automata theoretic technique can automatically check whether 
$S_1 \leq_{dom}^C S_2$. We omit the details of this technique here, which is presented in Appendix \ref{sec:toolusage} Proposition \ref{prop:checkMustDominance}. 
This technique is implemented in our tool \DCSYNTH\/. In case $S_1 \leq_{dom}^C S_2$ does not hold, the tool provides a counter example.
 
For our case studies, we experimentally  compare must dominance of supervisors 
$\GODSC_i(A,C)$ as defined in Lemma \ref{lemma:mustdom}. Recall that $\GODSC_i(A,C)$ denotes the maximally permissive $H$-optimal supervisor for the specification $Type_i(A,C)$.
The results obtained (with $H=50$) are as follows.
\begin{enumerate}
\item Mine-pump instance $Minepump(8,2,6,2)$ denoted by $MP(8,2,6,2)$): 
{\scriptsize    \begin{itemize}
        \item [] $\GODSC_1(MP(8,2,6,2)) ~<_{dom}^C~ \GODSC_3(MP(8,2,6,2)) ~=_{dom}^C~ \GODSC_2(MP(8,2,6,2))$
    \end{itemize}
    }
\item  Arbiter instance $Arb(5,3,2)$: 
  {\scriptsize   \begin{itemize}
        \item []  $\GODSC_1(Arb(5,3,2)) ~<_{dom}^C~ \GODSC_2(Arb(5,3,2)) ~=_{dom}^C~  \GODSC_3(Arb(5,3,2))$
   \end{itemize}
   }
\end{enumerate}
$\GODSC_3$ must dominates $\GODSC_1$ as expected, as $\GODSC_3$ is a sub-supervisor of $\GODSC_1$. What is interesting and surprising is that
in both the case studies Arbiter and Mine-pump,  the $\GODSC_2$ and $\GODSC_3$ supervisors
are found to be syntactically identical. This is not theoretically guaranteed, as $Type2$ and $Type3$ supervisors are must-incomparable in general. 
Thus, in these examples,
the $H$-optimal $\GODSC_2$ already provides all the must-guarantees of the hand-crafted $\GODSC_3$ hard requirements. The $H$-optimization of 
$C$ seems to exhibit startling ability to  guarantees $C$ without human intervention. It will be our attempt to validate this with more examples in future. So far we have considered commitment as soft requirement.
 In general, the soft requirement can be used to optimize \MPNC w.r.t. any regular property of interest, where as the hard requirements gives the necessary must guarantees. Such soft requirements may embody performance and quality goals. Hence, it is advisable to use the combination of hard and soft requirement based on the criticality of each requirement. 

\vspace*{-0.3cm}
\section{Discussion along with Related Work}
\label{section:discussion}
 \vspace*{-0.3cm}
 Reactive synthesis from Linear Temporal Logic (LTL) specification is a widely studied area\cite{BCGHHJKK14} and a considerable number of tools \cite{BBFJR12, FFT17} supported by theoretical foundations are available. The leading tools such as Acacia+\cite{BBFJR12} and BoSy\cite{FFT17} mainly focus on the future fragment of LTL. 
In contrast, this paper focuses on {\em invariance} of complex regular properties, denoted by $\invariant~D^h$ where $D^h$ is a \qddc\/ formula. For such a property, a maximally permissive supervisor ($\MPNC$) can be synthesized. Formally, logics LTL and \qddc\/ have incomparable expressive power. There is increasing evidence that regular properties form an important class of requirements \cite{DLB14, LRT17,MPW17}. The IEEE standard PSL extends LTL with
regular properties \cite{PSL}. Wonham and Ramadge in their seminal work \cite{RW87,RW89} first studied the synthesis of maximally permissive supervisors from regular properties.
In  their supervisory control theory, $\MPNC$ can in fact be synthesized for a richer property class $AGEF ~D^h$ 
 \cite{ELTV17}. Tool \DCSYNTH can be easily extended to support such properties too. 
Riedweg {\em et al} \cite{RP03} give some sub-classes of Quantified Mu-Calculus for which $\MPNC$ can be computed. However, none of these works address soft requirement guided synthesis. 

Most of the reactive synthesis tools focus on correct-by-construction synthesis from hard requirements.
For example, none of the tools in recent competition on reactive synthesis, SYNTCOMP17 \cite{SYN17},  address the issue of guided synthesis which is our main focus. In our approach, we refine the \MPNC\/ (for hard requirements) to a sub-supervisor optimally satisfying  the soft requirements too. 
Since LTL does not admit \MPNC, it is unclear how our approach can extend to it.

In quantitative synthesis, a weighted arena is assumed to be available, and algorithms for  optimal controller synthesis for diverse objectives such as Mean-payoff \cite{BCHJ09} or energy \cite{BMRLL18} have been investigated. In our case, we first synthesize the weighted arena from given hard and soft requirements. Moreover, we use $H$-optimality as the synthesis criterion. This criterion has been widely used in reinforcement learning as well as optimal control of MDPs \cite{Put94,Bel57}.
In other related work, techniques for optimal controller synthesis are discussed by Ding {\em et al} \cite{DLB14}, Wongpiromsarn {\em et al} \cite{WTM12} and Raman {\em et al} \cite{VADRS15}, where they have explored the use of receding horizon model predictive control along with temporal logic properties. 
\oomit{
 Significant advancement has been made in the area of LTL synthesis to produce better quality controllers based on mean-pay off objective. But, it requires either the quantitative properties (for guidance) specified as another DFA or a specification, where each alphabet can be given a weight \cite{BCHJ09, BBFR13}. It makes these techniques difficult to use in practice. However, in our framework, a logic \qddc can be used uniformly for specification (using hard requirement) as well as  guidance (using soft requirement). Further, our technique uses soft requirements ($D^s$) to synthesize $H$-optimal supervisor from $MPNC$.
}

Since our focus is on the quality of the controllers, we have also defined metrics and measurement techniques for comparing the controllers for their guaranteed (based on must dominance) and expected case performance. For the expected case measurement, we have assumed that inputs are $iid$. However, the method can easily accommodate a finite state Markov model governing the occurrences of inputs. 

\DCSYNTH\/ uses an efficient BDD-based symbolic representation, inherited from tool MONA \cite{Mon01} for storing automata, supervisors and
controllers. The use of eager minimization (see Appendix \ref{sec:ssdfa} for implementation details) allows us to handle much more complex properties (see  Appendix \ref{appendix:comparision}).

\vspace{-2ex}
\section{Conclusions}
\label{section:conclusion}
\vspace{-1ex}
We have presented a technique for guided synthesis of controllers from hard and soft requirements specified in logic \qddc. 
This technique is also implemented in our tool \DCSYNTH\/.
Case studies show that combination of hard and soft requirements provides us with a capability to deal with unrealizable (but desirable), conflicting and default requirements. 
In context of assumption-commitment based specification, we have shown with case studies that 
soft requirements improve the expected case performance, where as  hard requirements provide certain (but typically conditional) guarantees on the synthesized controller. Hence, the combination of hard and soft requirements as formulated in $Type3$ specifications offers a superior choice of controller specification. This is confirmed by theoetical analysis as well as experimental results. 
\oomit{
The tool gives us a Mealy machine (with non-deterministic outputs) offering various realizable options to choose from.  This is a novel tool, which has the guidance as a part of specification itself.
}
In the paper, we have also explored the experimental ability to compare the controller performance using \emph{expected value} and \emph{must dominance} metrics. This helps us in designing better performing controllers.

\bibliographystyle{plain}
\bibliography{awRef1}
\clearpage
\appendix
\section{Other Case Studies}
\label{appendix:morecasestudies}

In this section we present 3 more case studies: 
\begin{enumerate}
\item $n$-client shared resource arbiter (several different specifications),
\item alarm annunciation system.
\end{enumerate}

\subsection{Arbiter}
\label{sec:casestudy_arbiter}

An $n$-client resource (e.g. bus) arbiter is a circuit 
with $r_1,\ldots,r_n$ as inputs ($r_i$ high indicates $i^{th}$ client request access to resource) and 
$ack_1,\ldots,ack_n$ ($a_i$ indicates arbiter has granted the $i^{th}$ client access to the resource) as the corresponding outputs. Arbiter arbitrates among a subset of requests 
at each cycle by setting one of the acknowledgments ($a_i$'s) {\em true}. 
Hard requirements on the arbiter include the following three invariant properties.
\begin{equation}
 \begin{array}{l}
  Mutex(n) \df true\textrm{\textasciicircum} \langle ~\land_{i \neq j} ~\neg ( a_i \land a_j) ~ \rangle, ~1 \leq i \leq n \\
  NoLoss(n) \df true\textrm{\textasciicircum} \langle ~~ (\lor_i r_i) \Rightarrow (\lor_j a_j) ~\rangle, ~1 \leq i \leq n \\
  NoSpurious(n) \df true\textrm{\textasciicircum} \langle ~~\land_i ~(a_i \Rightarrow r_i) ~~\rangle, ~1 \leq i \leq n \\
  ARBINV(n) \df Mutex(n) \land NoLoss(n) \land NoSpurious(n). 
 \end{array}
\end{equation}
Thus, $Mutex$ gives mutual exclusion of acknowledgments, $NoLoss$ states that
if there is at least one request then there must be  an acknowledgment and 
$Nospurious$ states that acknowledgment is only given to a requesting cell. 

In the literature various arbitration schemes for the arbiter have been proposed, 
here we consider the following schemes. 
\begin{itemize}
\item $k$-cycle response time: 
let $Resp(r,a,k)$ denote that if request has been high for last $k$ cycles 
there must have been at least one acknowledgment in the last $k$ cycles. 
Let $ArbResp(n,k)$ state
that for each cell $i$ and for all observation intervals the formula $Resp(r_i,a_i,k)$ holds. 
\begin{equation}
\begin{array}{l}
  Resp(r,a,k) = 
true\textrm{\textasciicircum}(([[r]]\ \&\&\ (slen =(k-1)))~\Rightarrow \\
   \hspace*{2cm} true\textrm{\textasciicircum}(scount\ a > 0 \ \&\&\ (slen =(k-1)))\\
  ArbResp(n,k) \df (\land_{1 \leq i \leq n} ~(Resp(r_i,a_i,k)) \\
  ArbCommit(n,k) \df ARBINV(n) \land ArbResp(n,k)
\end{array}
\end{equation}

Based on $k$-cycle response we can define various arbiter specification with different properties as follows:
\begin{itemize}

\item Then specification $Arb^{hard}(n,k)$ is the following $k$-cycle response time \DCSYNTH specification. 
\begin{small}
\begin{equation}
 Arb^{hard}(n,k) \df (\{r_1, \ldots, r_n\}, \{a_1, \ldots, a_n\}, ArbCommit(n, k), \langle \rangle) 
\end{equation}
\end{small}

\item The specification $Arb^{hard}$ above, invariantly satisfy $ArbCommit(n,k)$ if $n \leq k$. Tool  \DCSYNTH gives us
a concrete controller for the instance  ($D^h=ArbCommit(6,6)$, $D^s=true$).
It is easy to see that there is no controller which can invariantly satisfy $ArbCommit(n,k)$ if $k < n$.
Consider the case when all requests $r_i$ are continuously true. Then, it is not possible to give response to
every cell in less than $n$ cycles due to mutual exclusion of acknowledgment $a_i$. 

To handle such desirable but unrealizable requirement we make an assumption. Let the proposition
$Atmost(n,i)$ be defined as \\
$\forall S \subseteq  \{1 \ldots n\}, |S| \leq i. ~~  \land_{j \notin S} \neg r_j$.\\ It states that  at most $i$ requests are true simultaneously. Then, the {\bf arbiter assumption} is the formula
{\bf $ArbAssume(n,i)$} = \verb#[[ Atmost(n,i)]]#, which states that $Atmost(n,i)$  holds invariantly in past.

The synchronous arbiter specification {\bf $Arb(n,k,i)$} is the assumption-commitment pair $(ArbAssume(n,i), ArbCommit(n,k))$. The four types of controller specifications can be derived from this pair.
Figure \ref{figure:arbiterspecification} in Appendix \ref{sec:toolusage}  gives, in textual syntax of the specification for
$\TYPE 3(Arb(5,3,2))$, in tool \DCSYNTH. The \DCSYNTH specification for $\TYPE 3(Arb(n,k,i))$ is denoted by $Arb^{hardAssume}$ given as follows:
\begin{small}
\begin{equation}
 Arb^{hardAssume}(n,k,i) \df 
 \begin{array}[t]{l}
 (\{r_1, \ldots, r_n\}, \{a_1, \ldots, a_n\}, \\
 \mathit{pref}(ArbAssume(n,i)) \Rightarrow EP(ArbCommit(n,k)), \\
 ~~\langle ArbCommit(n,k) \rangle~) 
 \end{array}
\end{equation}
\end{small}

\item $k$-cycle response time as soft requirement: 
we specify the requirement of response in $k$ cycles as a 
\emph{soft requirement} \footnote{Note that soft requirement in this example is a lexicographical list of several \qddc formulas. The tool \DCSYNTH implements a \GODSC computation using weighted requirements (See \ref{subsec:computegodc} for details)} as below. 
\begin{small}
\begin{equation}
 Arb^{soft}(n,k) \df 
 \begin{array}[t]{l}
 (\{r_1, \ldots, r_n\}, \{a_1, \ldots, a_n\}, ARBINV(n), \\
 ~~\langle Resp(r_n,a_n,k):2^n,\ldots, Resp(r_1,a_1,k):2^1 \rangle~) 
 \end{array}
\end{equation}
\end{small}

\end{itemize}
\item Token ring arbitration: 
a token is circulated among the masters in a round robin fashion. 
The token is modeled using the variables $tok_i$'s ($1\leq i\leq n$). 
Exactly one of $tok_i$'s will hold at any time and 
if $tok_i$ is {\em true} then we mean that $master_i$ holds the token. 
The arbiter asserts acknowledgement $a_i$ whenever request $r_i$ and $tok_i$ are {\em true}, 
i.~e.~priority is accorded to the request of the master which holds the token. 
\begin{equation}
\begin{array}{l}
TokInit(n)\df <tok_1\ \&\&\ (\land_{2\leq i\leq n}!tok_i)> \textrm{\textasciicircum}\ true\\
TokCirculate(n)\df [](\land_{1\leq i\leq n}({{tok_i}}\ \textrm{\textasciicircum}\ (slen =1) <=> \\
\hspace{6cm}(slen =1)\ \textrm{\textasciicircum}\ {{tok_{i\%n+1}}}))\\
TokResp(n)\df \land_{1\leq i\leq n}[[(r_i\ \&\&\ tok_i) => a_i ]]\\
Token(n)=TokInit(n)\land TokCirculate(n)\land TokResp(n).\\
\end{array}
\end{equation}
Let $ARBTOKEN(n)=ARBINV(n)\land Token(n)$. 
Then $Arb^{tok}(n)$ is the following \DCSYNTH\/ specification. 
\begin{small}
\begin{equation}
 Arb^{tok}(n) \df (\{r_1, \ldots, r_n\}, \{a_1, \ldots, a_n\}, ARBTOKEN(n), \langle \rangle ,\langle \rangle) 
\end{equation}
\end{small}
\end{itemize}

\subsection{Alarm Annunciation System}
\label{subsection:aas}
The next case study is Alarm Annunciation System(AAS) used in a process control system for 
annunciation for various alarms in the control room. 
The Alarm Annunciation involves the standard \emph{Automatic Ring-Back Sequence}
for all the digital inputs meant for alarm annunciation and provide the necessary outputs.
The specification of Automatic Ring-Back Sequence is given in Table \ref{tab:aas}.
All digital inputs representing alarm conditions are scanned periodically. 

\begin{table}[h!]
  \begin{center}
    \caption{Automatic Ring-Back Sequence for Alarm Annunciation}
    \label{tab:aas}
    \begin{tabular}{|p{0.2\linewidth} | p{0.4\linewidth} | p{0.4\linewidth}|}
    \hline 
      \textbf{Input} & \textbf{Lamp Output} & \textbf{Audio Output}\\
      \hline
      Normal to Alarm & Fast Flashing & Normal Alarm Hooter On\\
      \hline
      Acknowledged & Lamp On & Normal Alarm Hooter Off\\
      \hline
      Alarm to Normal & Slow Flashing & Ringback Hooter On\\
      \hline
      Reset & Lamp Off & Ringback Hooter Off\\
      \hline
    \end{tabular}
  \end{center}
\end{table}

As shown in the Table \ref{tab:aas} that Automatic Ring-Back Sequence specification takes the alarm
signal as input. The high value of signal represents the alarm state, otherwise the signal is said to
be in normal state. Other inputs are Acknowledgment, Reset and Silence inputs, which are controller 
by the operator.
There are three output elements: Lamp, Normal Hooter and Ringback Hooter.
There is a Lamp corresponding to each alarm signal, whereas Hooters are common to all alarm signals.
Lamp can either be \emph{Fast Flashing, Slow Flashing, Steady On or Off} states.
We have encoded the requirements in \DCSYNTH\/ to synthesize the controller.
The Silence input can be used by the operator to switch off Hooters.

\subsubsection{Result of Synthesis}: As discussed in the previous case study soft requirements
helped in specification of requirements concisely. The controller synthesized has 8 states.
We could simulate the controller and verified the correctness.

{\color{black}
\subsection{Synthesis of 2 client arbiter}
\label{section:2cellexample}
Fig.~\ref{fig:monitor2cell} gives the monitor automaton for 2-client arbiter (See section \ref{sec:casestudy_arbiter}) for n-cell arbiter specification.
 Each transition is labeled by 4 bit vector giving values of $r_1, r_2, a_1, a_2$.

\begin{figure}[!h]
\centering
\includegraphics[width=\textwidth, keepaspectratio]{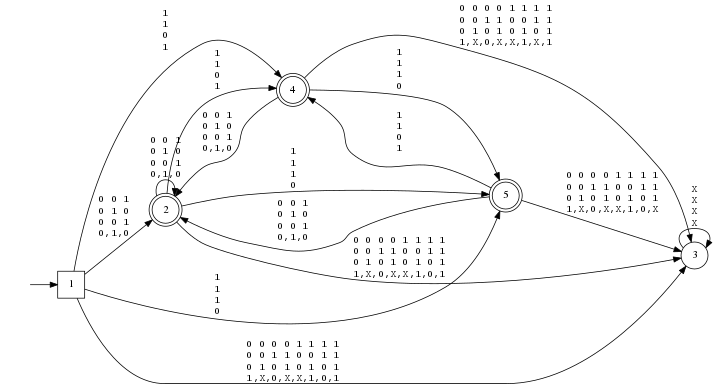}
\caption{Safety Monitor Automaton: 2 Cell Arbiter}
\label{fig:monitor2cell}
\end{figure}

\begin{figure}[h]
\begin{minipage}{2.5in}
\includegraphics[scale=.4,keepaspectratio]{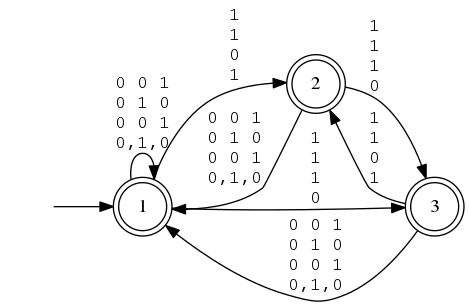}
\end{minipage} \ \ 
\begin{minipage}{2.5in}
\includegraphics[scale=.4,keepaspectratio]{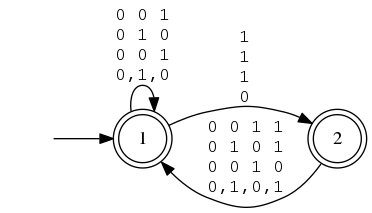}
\end{minipage}
\caption{Supervisors for $Arb^{hard}(2,2)$ 
(a): \MPNC  (b): \GODSC determinized}
\label{fig:mpnc}
\end{figure}

Fig.~\ref{fig:mpnc} gives the \MPNC automaton for the 2-cell arbiter computed from 
the safety monitor automaton of Fig.~\ref{fig:monitor2cell}. 
(There is an additional reject state. All missing transitions are
directed to it. These are omitted from the diagram for simplicity.) Note that this is a DFA 
whose transitions are labelled by 4-bit vectors representing alphabet $2^{\{r_1, r_2, a_1, a_2\}}$.
As defined in Definition \ref{def:nondetmm},
the DFA also denotes an output-nondeterministic Mealy machine with input variables $(r_1,r_2)$ and 
output variables $(a_1,a_2)$.
The automaton is nondeterministic in output 
as from state $1$, on input $(1,1)$ it can move to state $2$ with output $(1,0)$, 
or to state $3$ with output $(0,1)$. The reader can verify that the automaton is non-blocking and hence
a controller.

In 2-cell arbiter example, with soft requirements $\langle ack1, ack2 \rangle$ 
which give $ack1$ priority over $ack2$, 
we obtain the \GODSC controller automaton of Fig.~\ref{fig:mpnc}(b) from the \MPNC of Fig.~\ref{fig:mpnc}(a).
Note that we minimize the automaton at each step.
}
\section{Specification of Mine-pump and Arbiter example in \DCSYNTH}
\label{sec:exampleSources}
The specification of Arbiter(5,3,2) and MinePump(8,2,6,2) is \DCSYNTH syntax is given is Figure \ref{figure:minepumpspecification} and \ref{figure:arbiterspecification} respectively.
\begin{figure}[!b]
\caption{Mine-pump specification in \DCSYNTH}
\label{figure:minepumpspecification}
\begin{scriptsize}

\framebox{\parbox[t][][t]{\columnwidth}{
$\begin{array}{l}
\mathrm{\textsf{\#qsf "minepump"}}\\
\mathrm{\textsf{interface\{}}\\
\quad\mathrm{\textsf{input HH2Op, HCH4p;}}\\
\quad\mathrm{\textsf{output PUMPONp monitor x, ga monitor x;}}\\
\quad\mathrm{\textsf{constant w = 8, epsilon=2 , zeta=6, kappa=2;}}\\
\mathrm{\textsf{\}}}\\
\mathrm{\textsf{definitions\{}}\\
\mathrm{\textsf{//Methane release assumptions}}\\
\mathrm{\textsf{dc methane1(HCH4)\{}}\\
\quad\mathrm{\textsf{[]([HCH4]\textasciicircum[!HCH4]\textasciicircum \textless HCH4 \textgreater =\textgreater slen\textgreater zeta );}}\\
\mathrm{\textsf{\}}}\\
\mathrm{\textsf{dc methane2(HCH4)\{}}\\
\quad\mathrm{\textsf{[]([[HCH4]] =\textgreater slen\textless kappa );}}\\
\mathrm{\textsf{\}}}\\
\mathrm{\textsf{//Pump capacity assumption}}\\
\mathrm{\textsf{dc pumpcap1(HH2O, PUMPON)\{}}\\
\quad\mathrm{\textsf{([]!(slen=epsilon \&\& ([[PUMPON \&\& HH2O]] \textasciicircum \textless HH2O \textgreater)));}}\\
\mathrm{\textsf{\}}}\\
\mathrm{\textsf{dc MineAssume\_2\_6\_2(HH2O, HCH4, PUMPON)\{}}\\
\quad\mathrm{\textsf{methane1(HH2O, HCH4, PUMPON) \&\& methane2(HH2O, HCH4, PUMPON) \&\&}}\\ 
\quad\mathrm{\textsf{pumpcap1(HH2O, HCH4, PUMPON);}}\\ 
\mathrm{\textsf{\}}}\\
\mathrm{\textsf{//safety condition}}\\ 
\mathrm{\textsf{dc req1(HH2O, HCH4, PUMPON)\{}}\\
\quad\mathrm{\textsf{true\textasciicircum\textless( (HCH4 $||$ !HH2O) =\textgreater !PUMPON)\textgreater;}}\\
\mathrm{\textsf{\}}}\\
\mathrm{\textsf{dc req2(HH2O, HCH4, PUMPON)\{}}\\
\quad\mathrm{\textsf{(!(true \textasciicircum ([[HH2O]] \&\& (slen = w))))}}\\
\mathrm{\textsf{\}}}\\
\mathrm{\textsf{dc MineCommit\_8(HH2O, HCH4, PUMPON)\{}}\\
\quad\mathrm{\textsf{req1(HH2O, HCH4, PUMPON) \&\& req2(HH2O, HCH4, PUMPON);}}\\
\mathrm{\textsf{\}}}\\
\mathrm{\textsf{indefinitions\{}}\\
\quad\mathrm{\textsf{ga : MineCommit\_8(HH2Op, HCH4p, PUMPONp);}}\\
\mathrm{\textsf{\}}}\\
\mathrm{\textsf{hardreq\{}}\\
\quad\mathrm{\textsf{MineAssume\_2\_6\_2(HH2Op, HCH4p, PUMPONp) =\textgreater}} \\
\quad\quad\mathrm{\textsf{MineCommit\_8(HH2Op, HCH4p, PUMPONp);}}\\
\mathrm{\textsf{\}}}\\
\mathrm{\textsf{softreq\{}}\\
\quad\mathrm{\textsf{useind ga;}}\\
\quad\mathrm{\textsf{(ga);}}\\
\mathrm{\textsf{\}}}\\
\end{array}$}
}
\end{scriptsize}
\end{figure}

\oomit{
\clearpage
\textbf{Simulation of Synthesized Mine pump Controllers}
The controllers are encoded as Lustre specification and Lustre V4 tools are used 
for simulation. The example simulation for these three variants of mine pump
with soft requirements MPV1, MPV2 and MPV3 are shown in figures 
~\ref{fig:pumptrue}, ~\ref{fig:pumpsoft} and ~\ref{fig:pumpfalse} respectively.

\begin{figure*}[!h]
\centering
\includegraphics[width=\textwidth, keepaspectratio]{PumpTrue.png}
\caption{Simulation of mine pump controller with soft requirement MPV1}
\label{fig:pumptrue}
\centering
\includegraphics[width=\textwidth, keepaspectratio]{PumpSoft.png}
\caption{Simulation of mine pump controller with soft requirement MPV2}
\label{fig:pumpsoft}
%
\centering
\includegraphics[width=\textwidth, keepaspectratio]{PumpFalse.png}
\caption{Simulation of mine pump controller with soft requirement MPV3}
\label{fig:pumpfalse}
\end{figure*}
}

\begin{figure}[!b]
\caption{Arbiter specification in \DCSYNTH}
\label{figure:arbiterspecification}
\begin{scriptsize}
\framebox{\parbox[t][][t]{\columnwidth}{
$\begin{array}{l}
\mathrm{\textsf{\#qsf "arbiter"}}\\
\mathrm{\textsf{interface\{}}\\
\quad\mathrm{\textsf{input r1, r2, r3, r4, r5;}}\\
\quad\mathrm{\textsf{output a1, a2, a3, a4, a5, ga3;}}\\
\quad\mathrm{\textsf{constant n=3;}}\\
\mathrm{\textsf{\}}}\\
\mathrm{\textsf{definitions\{}}\\
\mathrm{\textsf{// Specification 1: The Acknowlegments shold be exclusive}}\\
\mathrm{\textsf{dc exclusion()\{}}\\
\quad\mathrm{\textsf{true\textasciicircum\textless (a1 =\textgreater !(a2 $||$ a3 $||$ a4 $||$ a5)) \&\& (a2 =\textgreater !(a1 $||$ a3 $||$ a4 $||$ a5)) \&\& }}\\
\quad\mathrm{\textsf{(a3 =\textgreater !(a1$||$a2$||$a4$||$a5)) \&\& (a4=\textgreater !(a1$||$a2$||$a3$||$a5)) \&\& (a5=\textgreater !(a1$||$a2$||$a3$||$a4)))\textgreater;}}\\
\mathrm{\textsf{\}}}\\
\mathrm{\textsf{dc noloss()\{}}\\
\quad\mathrm{\textsf{true\textasciicircum\textless (r1 $||$ r2 $||$ r3 $||$ r4 $||$ r5) =\textgreater (a1 $||$ a2 $||$ a3 $||$ a4 $||$ a5)\textgreater; }}\\
\mathrm{\textsf{\}}}\\
\mathrm{\textsf{//If bus access (ack) should be granted only if there is a request}}\\
\mathrm{\textsf{dc nospuriousack(a1, r1)\{}}\\
\quad\mathrm{\textsf{true\textasciicircum\textless (a1) =\textgreater (r1)\textgreater; }}\\
\mathrm{\textsf{\}}}\\
\mathrm{\textsf{//n cycle response i.e. slen=n-1}}\\
\mathrm{\textsf{dc response(r1,a1)\{}}\\
\quad\mathrm{\textsf{true\textasciicircum (slen=n-1 \&\& [[r1]])  =\textgreater true\textasciicircum(slen=n-1 \&\& (scount a1 \textgreater= 1)); }}\\
\mathrm{\textsf{\}}}\\
\mathrm{\textsf{dc ArbAssume\_5\_2()\{}}\\
\quad\mathrm{\textsf{[[ (!r1 \&\& !r2 \&\& !r3 \&\& !r4 \&\& !r5) $||$ (r1 \&\& !r2 \&\& !r3 \&\& !r4 \&\& !r5) $||$ }}\\
\quad\mathrm{\textsf{(!r1 \&\& r2 \&\& !r3 \&\& !r4 \&\& !r5) $||$ (!r1 \&\& !r2 \&\& r3 \&\& !r4 \&\& !r5) $||$  }}\\
\quad\mathrm{\textsf{(!r1 \&\& !r2 \&\& !r3 \&\& r4 \&\& !r5) $||$ (!r1 \&\& !r2 \&\& !r3 \&\& !r4 \&\& r5) $||$  }}\\
\quad\mathrm{\textsf{(r1 \&\& r2 \&\& !r3 \&\& !r4 \&\& !r5) $||$ (r1 \&\& !r2 \&\& r3 \&\& !r4 \&\& !r5) $||$  }}\\
\quad\mathrm{\textsf{(r1 \&\& !r2 \&\& !r3 \&\& r4 \&\& !r5) $||$ (r1 \&\& !r2 \&\& !r3 \&\& !r4 \&\& r5) $||$  }}\\
\quad\mathrm{\textsf{(!r1 \&\& r2 \&\& r3 \&\& !r4 \&\& !r5) $||$ (!r1 \&\& r2 \&\& !r3 \&\& r4 \&\& !r5) $||$  }}\\
\quad\mathrm{\textsf{(!r1 \&\& r2 \&\& !r3 \&\& !r4 \&\& r5) $||$ (!r1 \&\& !r2 \&\& r3 \&\& r4 \&\& !r5) $||$  }}\\
\quad\mathrm{\textsf{(!r1 \&\& !r2 \&\& r3 \&\& !r4 \&\& r5) $||$ (!r1 \&\& !r2 \&\& !r3 \&\& r4 \&\& r5) ]];  }}\\

\mathrm{\textsf{\}}}\\

\mathrm{\textsf{dc guranteeInv()\{}}\\
\quad\mathrm{\textsf{exclusion() \&\& noloss(HH2O, HCH4, PUMPON) \&\& nospuriousack(a1,r1) \&\&}}\\ 
\quad\mathrm{\textsf{nospuriousack(a2,r2) \&\& nospuriousack(a3,r3) \&\& nospuriousack(a4,r4) \&\&}}\\ 
\quad\mathrm{\textsf{nospuriousack(a5,r5);}}\\ 
\mathrm{\textsf{\}}}\\

\mathrm{\textsf{dc guranteeResp()\{}}\\
\quad\mathrm{\textsf{response(r1,a1) \&\& response(r2,a2) \&\& response(r3,a3) \&\&}}\\ 
\quad\mathrm{\textsf{response(r4,a4) \&\& response(r5,a5);}}\\ 
\mathrm{\textsf{\}}}\\

\mathrm{\textsf{dc ArbCommit\_5\_3()\{}}\\
\quad\mathrm{\textsf{guranteeInv() \&\& guranteeResp() ;}}\\ 
\mathrm{\textsf{\}}}\\
\mathrm{\textsf{\}}}\\

\mathrm{\textsf{indefinitions\{}}\\
\quad\mathrm{\textsf{ga3 : ArbCommit\_5\_3();}}\\
\mathrm{\textsf{\}}}\\
\mathrm{\textsf{hardreq\{}}\\
\quad\mathrm{\textsf{ArbAssume\_5\_2() =\textgreater ArbCommit\_5\_3();}}\\
\mathrm{\textsf{\}}}\\
\mathrm{\textsf{softreq\{}}\\
\quad\mathrm{\textsf{useind ga3;}}\\
\quad\mathrm{\textsf{(ga3);}}\\
\mathrm{\textsf{\}}}\\
\end{array}$}
}
\end{scriptsize}
\end{figure}

\oomit{
\begin{verbatim}

m2l-str;
var2 r1, r2, r3, r4, r5, a1, a2, a3, a4, a5, ga3;
var1 i;

(
all2 a1N, a2N, a3N, a4N, a5N, ga3N:
((import("ArbHardAssume(5,3,2)/GODSC.dfa", 
       r1->r1, r2->r2, r3->r3, r4->r4, r5->r5, 
       a1->a1N, a2->a2N, a3->a3N, a4->a4N, a5->a5N, ga3->ga3N)
)=>  (i in ga3N) 
))
=>
(
(import("ArbHardAssumeSoft(5,3,2)/GODSC.dfa", 
       r1->r1, r2->r2, r3->r3, r4->r4, r5->r5, 
       a1->a1, a2->a2, a3->a3, a4->a4, a5->a5, ga3->ga3) 
)=>(i in ga3));

\end{verbatim}
}

\section{\DCSYNTH Tool Usage and performance evaluation}
\label{sec:toolusage}
The tool \DCSYNTH uses a specification for Arbiter \verb|Arbiter.qsf| shown in Figure \ref{figure:arbiterspecification}. This file contains
the set of input and output alphabets in \emph{interface} section. The definitions/macros required
for specifying hard and soft requirements are contained in
\emph{definitions} section.
This is followed by a section called \emph{indefinitions}, to specify the required indicating monitor for a given formula (or corresponding automaton). Finally the section called \emph{hardreq} and \emph{softreq} define the hard and soft requirements respectively using the definitions and indicating monitors. The steps to synthesize a controller from the specification file is as follows.
\begin{itemize}
\item First we generate the DFAs for $D^h$, $D^s$ and the required input/output partitioning file using \textbf{qsf} command, e.g. for Arbiter example, we use \textbf{qsf Arbiter.qsf} to generate files named \textbf{Arbiter.hardreq.dfa}, \textbf{Arbiter.softreq.dfa} and \textbf{Arbiter.io} as per step 1 of synthesis method in section \ref{section:dcsynth-spec}.
\item We then use the command \textbf{synth2 Arbiter.hardreq.dfa Arbiter.softreq.dfa Arbiter.io synth.config} to synthesize the supervisors as per step 2 and 3 of synthesis method in section \ref{section:dcsynth-spec}.
The file \textbf{synth.config} is used to provide the configuration parameters like the number of iterations for H-optimal supervisor. The command produces the supervisors \textbf{\MPNC.dfa} and \textbf{\GODSC.dfa}. 

\item We then determinize the \textbf{\GODSC.dfa} using default values with command \textbf{synth\_deterministic \GODSC.dfa default.io} to get a controller called \textbf{Controller.dfa}.
The file \textbf{default.io} contains the ordered list of output literals e.g. if we have two outputs o1 and o2, then the list \{!o1,o2\} says that try to determinize the \textbf{\GODSC.dfa} with following priority for the output \textbf{\{!o1,o2\}$>>$\{!o1,!o2\}$>>$\{o1,o2\}$>>$\{o1,!o2\}}.

\item To measure the expected value of the soft requirement being satisfied, the command \textbf{aut2mrmc Controller.dfa default.io index} is used. The \textbf{index} parameter is the index of indicator variable (for the soft requirement) in the \textbf{Controller.dfa}. The command produces \textbf{Controller.tra} and \textbf{Controller.lab} files which can be imported in tool MRMC to compute the expected value.

\item Apart from expected case performance, the tool also facilitates the method for checking \textit{must dominance} between two given supervisors $S_1$ and $S_2$. As supervisors are finite state mealy machines and a commitment $C$ is a regular property, we can use validity checking of \qddc to check whether  
$S_1 \leq_{dom}^C S_2$ as formulated in the following proposition.  In tool \DCSYNTH, we provide a facility to decide must-dominance between two supervisors. The tool also gives a  counter example if must-dominance fails.
\begin{proposition}
\label{prop:checkMustDominance}
 Let $\D(S_i)$ denote \qddc\/ formulas with same language as the supervisors $S_i$ and let $C$ be regular property over input-output alphabet $(I,O)$. Then,
 $S_1 \leq_{dom}^C S_2$ iff
 \[
 \models_{qddc} ~~ \forall I ~(\left[ \forall O.~~D(S_1) \land C)\right] \quad \Rightarrow \quad  \left[ \forall O. (D(S_1) \land C) \right])
\]
\end{proposition}

We use this facility to compare supervisors obtained from different types of specifications discussed in Section \ref{sec:controllertype}
\end{itemize}

\section{Synthesis with Semi-Symbolic DFA} 
\label{sec:ssdfa}
An interesting representation for total and deterministic finite state automata 
was introduced and implemented by Klarlund {\em et~al} in the tool MONA\cite{Mon01}. 
It was used to efficiently compute formula automaton for MSO over finite words. 
We denote this representation as {\em Semi-Symbolic DFA} (SSDFA). 
In this representation, the transition function is 
encoded as {\em multi-terminal BDD} (MTBDD).
The reader may refer to original papers \cite{Mon01,Mon02} for further details of MTBDD and 
the MONA DFA library. 

Here, we briefly describe the  SSDFA representation, and then consider controller synthesis on SSDFA. 
Figure \ref{fig:ssdfa}(a) gives an explicit DFA. 
Its alphabet $\Sigma$ is 4-bit vectors giving value of propositions $(r_1,r_2,a_1,a_2)$ 
and set of states $S=\{1,2,3,4\}$. 
This automaton has a unique reject state 4 and 
all the missing transitions are directed to it. 
(State 4 and transitions to it are omitted in Figure~\ref{fig:ssdfa}(a) for brevity.)
\noindent \begin{figure}[h]
\begin{minipage}{2in}
\includegraphics[scale=.3,keepaspectratio]{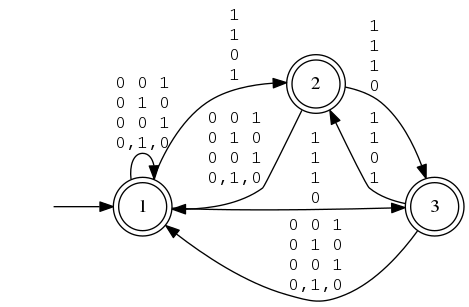}
\end{minipage} \ \ 
\begin{minipage}{3in}
\includegraphics[scale=.25,keepaspectratio]{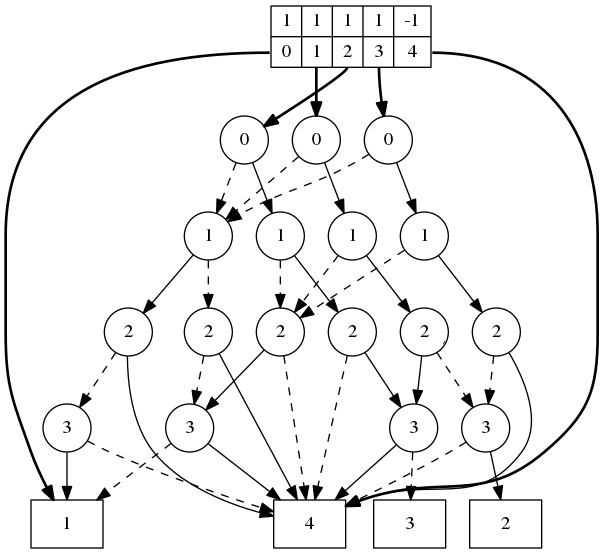}
\end{minipage}
\caption{$A^{\mpnc}$ for $Arb^{hard}(2,2)$ 
(a): External format  (b): SSDFA format}
\label{fig:ssdfa} 
\end{figure}

Figure \ref{fig:ssdfa}(b) gives the SSDFA for the above automaton. 
Note that states are explicitly listed in the array at top and 
final states are marked as $1$ and non-final states marked as $-1$. 
(For technical reasons there is an
additional state $0$ which may be ignored here and state 1 may be treated as the initial state). 
Each state $s$ points to
shared MTBDD node encoding the transition function $\delta(s):\Sigma \rightarrow S$ 
with each path ending in the next state. 
Each circular node of MTBDD represents a {\em decision node} 
with indices $0,1,2,3$ denoting variables $r_1,r_2,a_1,a_2$. Solid
edges lead to true co-factors and dotted edges to false co-factors. 

MONA provides a DFA library implementing automata 
operations including  product, complement, projection and minimization on SSDFA. Moreover, automata may be constructed
from scratch by giving list of states and adding transitions one at a time. 
A default transition must be given to make the automaton total. 
Tools MONA and DCVALID use eager minimization while converting formula into SSDFA.

\noindent {\bf Remark 1}: DFA in Figure \ref{fig:ssdfa} also denotes a 
Output-nondeterministic Mealy machine with input alphabet $(r_1,r_2)$ 
and output alphabet $(a_1,a_2)$. 
Automaton is nondeterministic in its output as $\delta(1,(1,1,1,0))$ $=2$ and 
$\delta(1,(1,1,0,1))=3$.

We use SSDFA to efficiently synthesize the \MPNC and \GODSC for the \DCSYNTH specification  $(I,O,D^h, D^s)$, without actually expanding the specification automata into game graph. 
The use of SSDFA leads to significant improvement in the scalability and computation time of the tool.

\subsection{Computing Maximally Permissive Supervisor (\MPNC)}
\label{subsec:computempnc}
Recall the synthesis method in Section \ref{section:dcsynth-spec}.
Let the hard requirement automaton be $A(D^h)=\langle S,2^{I\union O},\delta,F\rangle$.
We construct the  maximally permissive supervisor
by iteratively applying $Cpre(A(D^h), X)$ to compute set of winning states $G$, as outlined section \ref{sec:MPNCConstruction}.
This requires efficient implementation of $\mathit{Cpre(A(D^h),X)}$ over SSDFA $A(D^h)$.  The symbolic algorithm for $\mathit{Cpre}$  marks, \textbf{(a)} each leaf node representing state $s$ by truth 
value of $s \in X$, \textbf{(b)} each decision node associated with an input variable with $AND$ of its children's value, and \textbf{(c)} each decision node associated with output variable with $OR$ of its children's value. The computation is carried out bottom up on MTBDD and takes time $|MTBDD|$, 
where $|MTBDD|$ is the number of BDD nodes in it. 
In contrast the enumerative method for implementation of $Cpre$ would have taken time of the order of $2^{| I \cup O|}$.

Next we compute the automaton $\MPNC(D^h) = \langle G\union\{r\}, 2^{I\union O},G,\delta'\rangle$ 
by only retaining transitions between the winning states $G$. Here $r$ is the unique reject state introduced to make the automaton total.
We consider the following two methods. 
\begin{itemize} 
\item {\em Enumerative method:} $\MPNC(D^h)$ is constructed from $\A(D^h)$ by adding a transition at a time as follows: 
for any $s \in G$ if $\delta(s,(i,o))\in G$ then $(s,(i,o),\delta(s,$ 
$(i,o))\in\delta'$.
Clearly, this algorithm has  time complexity $|S| \times 2^{| I \cup O|}$.
Finally, we make $A^{\mpnc}$ total by adding all the unaccounted transitions from any state to the reject state $r$.
\item {\em Symbolic method:}  in this method, the
MTBDD of $\A(D^h)$ is modified so that each edge pointing 
to a state in $S-G$ is changed to go to the reject state $r$. 
Note that this makes states in $S-(G \cup \{r\})$ inaccessible. 
Now this modified SSDFA is minimized to get rid of inaccessible states and to get smaller \MPNC. 
The time complexity of this computation is $O(|MTBDD|)$ for modifying the links 
and $N.t.log(N)$ for minimization where $N$ is number of states and
$t$ is the size of alphabet in $A(D^h)$.
 \end{itemize}
In {\bf Table \ref{tab:comparisionmpncopt}} we give experimental results comparing the computation of $\MPNC(D^h)$ using the two algorithms. 
It can be seen that the symbolic algorithm can be faster by several orders of magnitude.
This is because we do not construct the \MPNC from scratch; instead we only redirect some links in MTBDD of $A(D^h)$ which is already computed.
The Mine-pump specification used in the  Table \ref{tab:comparisionmpncopt} is given in Section \ref{section:minepumpcasestudy}.



\begin {table}[!h]
\caption {\MPNC Synthesis: Enumeration vs symbolic method (time in seconds). For $\A(D^h)$ we give
number of states and time to compute it from the $\qddc$ hard requirement formula. For $\MPNC(D^h)$ we give
its number of states and time to compute it using the two methods.
$S_t$, $T_s$, $En$ and $Sy$ represent total no. of states, time in seconds, enumerative method and symbolic method respectively.
The Example $Arb^{hard}(n,k)$ represents the specification $\TYPE 0(Arb(n,k,n))$ and $Arb^{soft}(n,k)$ represents the example $\TYPE 2(Arb(n,k,n))$.
}
\label{tab:comparisionmpncopt}
        \begin{center}
        \begin{tabular}
        {|c|c|c|c|c|c|c|}
                \hline
                \multirow{2}{*}{Example} &\multirow{2}{*}{Hard Requirement} &  \multicolumn{2}{|c|}{ $\A(D^h)$ }    & \multicolumn{3}{|c|}{ $\MPNC(D^h)$ }\\
                \cline{3-7}
                && $S_t$ & $T_s$   & $S_t$ & $T_s$ & $T_s$  \\
                &&  &   &  & ($En$) & ($Sm$) \\
                
                \hline      
                $Arb^{hard}(4,4)$& $ARBHARD(4,4)$ & 177 & 0.04 & 126  & 0.025563 & 0.002033 \\
                \hline    
                $Arb^{hard}(5,5)$& $ARBHARD(5,5)$ & 2103  & 0.43 &  1297 & 0.59 & 0.04 \\
                \hline    
                $Arb^{hard}(6,6)$& $ARBHARD(6,6)$ & 31033  & 9.22 & 16808  & 42.75 & 0.91 \\
                \hline    
                $Arb^{soft}(4,2)$ & $ARBINV(4)$ & 3 & 0.016 &  2 & 6.2E-4 & 7.6E-5 \\
                \hline
                $Arb^{soft}(5,3)$ & $ARBINV(5)$ & 3  & 0.020 &  2 & 1.9E-3  & 1.2E-4 \\
                \hline
                $Minepump(8,2,10,1)$ & $MineAssume(2,10,1)$  & 271 &  0.08 & 211  & 1.4E-2 & 5.8E-3 \\
            
                & $\Rightarrow ArbCommit(8)$  &  &   &   &  &  \\
                \hline  
           
	\end{tabular}
\end{center}
\end{table}

\subsection{Computing $H$-Optimal Supervisor (\GODSC)}
 \label{subsec:computegodc}
 
In this step we compute the \GODSC from \MPNC. For a given maximally permissive supervisor \MPNC, a \qddc formula $D$ and an integer parameter H. We get the H-optimal sub-supervisor of \MPNC called \GODSC by iteratively computing  $Val(s,p+1)$ from $Val(s,p)$ for $0 \leq p < H$ as outlined in Section \ref{sec:HoptimalComputation} \footnote{Note that the tool \DCSYNTH in general allows a lexicographical list of soft requirement, it is basically a lexicographical list of several \qddc formulas. The tool \DCSYNTH implements a \GODSC computation based on this lexicographical (or with explicit weight to each soft requirement) list by using the weight for each transition as the \emph{sum of weights of all the soft requirement} being satisfied on that transition. The tool also allows the discounting factor $\gamma$ which is used to give higher weight to the requirements being satisfied in near future}.
This step can be denoted by  $\mathit{VALPre(\A^{Arena})}$.

{\bf Remark 2}: For $\A^{Arena}$  a transition has the form $\delta(s,(i,o,v))$ with $i \in 2^I, o \in 2^O, v \in 2^{\{w\}}$. However, from the definition of $Ind(D^s, w)$, the value of $w$  is uniquely determined by $(s,(i,o))$ in the corresponding automaton $\A(Ind(D^s, w))$. Hence we can abbreviate the transition as $\delta(s,(i,o))$. \qed
 
Now to compute $\GODSC$ we again have two methods: one is enumerative and other is symbolic method. We give the algorithm and associated complexity results for one value iteration (i.e. for $\mathit{VALpre}$  followed by $O_{max}$ computation. 
 Let $Q$ be the set of states of $\A^{Arena}$.

\begin{itemize}
\item \textit{Enumerative Method}: As given in Step (3) of synthesis method, for each state $s$ we need
to enumerate all paths starting from $s$ to get $Val(s,p+1)$ from $Val(s,p)$, which will take time of the
order of $2^{|I \cup O|} \times k$, where $k$ is the number of soft requirements (In this paper $k$ is assumed to be 1). 
Similar complexity will be required to get the list of transitions with maximum values denoted as $o_{max}$ (Note that there can be multiple transitions with same $o_max$, all such transitions will be included in \GODSC).
Hence, As the algorithm terminates after $H$ iterations the total time complexity of entire algorithm for $H$ iteration is $|Q| \times 2^{| I \cup O|} \times H$ (for $k=1$). 

\oomit{
 \item {\em Enumerative method (for $i^{th}$ iteration):} For each state $s \in Q$ and for each input $i \in 2^I$ compute $o_{max}^i$ and corresponding target state $s'$ as explained above. 
We add the transition $\delta_{godsc}(s,(i,o_{max}))=(s',t')$ to $A^{godsc}$.  
This takes $|Q| \times 2^{|I|}$ computations of $o_{max}^i$. 
For finding $o_{max}^i$ for a given input, 
we enumerate each possible output maintaining the value of current optimal $o_{max}^i$. 
Finding $o_{max}^i$ takes $2^{|O|} \times k$ steps (each of constant time) where $k$  is the number of 
soft goals. This is because evaluating $ValSoft$ takes ${\mathcal O}(k)$ time. 
Hence the time complexity of entire algorithm shall be $|Q| \times 2^{| I \cup O|} \times k$.
}
 
\item {\em Symbolic method:} 
For this optimization to be applicable
we assume that in MTBDD representation of $A^{Arena}$, all the input variables occur before the output variables $O$ and the indicating variable $w$ (in general it can be a set if $k>1$).
A node in MTBDD is called a {\em frontier node} if it is labelled with an output or a witness variable, and all its ancestors are labelled with input variables. 
For example, in Figure \ref{fig:ssdfa}(b), these are nodes labelled 2 (they happen to occur at same level in this example).
For each frontier node enumerate each path $\pi$ within the MTBDD below the frontier node 
(this fixes values of $(o,w) \in 2^O \times  2^W$ occurring on $\pi$ as well as next state $s'$). 
Update the optimal $o_{max}$  as well as next state $s_o$  based on $wt(o,v)$ for paths seen so far. 
This takes time $O(d_f \times k)$ where  $d_f$ is the number of paths 
in $MTBDD$ below the frontier node $f$. 
This optimal output $o_{max}(f)$ as well as next state $s_o$ for \emph{each value iteration} is  stored in each frontier node $f$. 
The total time taken is $O(d_{output}  \times H)$ where $d_{output} = \Sigma_{f \in Fr} ~ d_f$ and $Fr$ is the set of all frontier nodes. $H$ is the number of steps in value iterations. 

In second step, for each state $s \in Q$, enumerate each path from state $s$ to a frontier node $f$. 
This fixes the valuation of input $x$. Insert a transition $\delta_{\godsc}(s,(x,o_{max})) = s_o$ to $\A^{\godsc}$. Let the total number of paths up to frontier nodes be $d_{input}$. 
Then the second step takes time $O(d_{input} + |Q|)$ where time taken to insert a transition in $\A^{\godsc}$ is assumed to be constant. Hence total time for entire algorithm is $\A^{\godsc}$ is $O(d + |Q|) \times H$ where $d$ is total number of paths in MTBDD of $\A^{Arena}$ (here $k$ is assumed to be 1).
\end{itemize}


It may also be noted that in worst case, the total number of MTBDD paths $d$ is of 
size $O(2^{|I \cup O|})$ and two algorithms have comparable complexity. 
But in most cases, the total number of MTBDD paths $d \ll 2^{|I \cup O|}$ 
and the symbolic algorithm turns out to be more efficient.

\subsection{Computing a Controller using default value}
 \label{subsec:computeController}
 The controller can be computed from \GODSC for a given default value order $ord$, using the similar algorithm as given for \GODSC computation. Here we assume that the default values provided are the soft requirements and $H$ is equal to 1. So the \GODSC computation algorithm will try to choose those transitions with outputs that locally satisfy the default values given in $ord$.
 As the $ord$ provides default values for every output variable, so there will always be a unique output that will maximally satisfy the default value. Hence, the output will always be a deterministic Mealy machine i.e., we will get a controller.
\section{Comparison with Other tools}
\label{appendix:comparision}
In Table \ref{tab:comparision} we have compared the performance of DCSynth with few leading tools
for LTL synthesis.
The examples in QDDC are
manually translated into bounded LTL properties for giving them as input to Acacia+ \cite{BBFJR12} and BoSy \cite{FFT17}.
We have only considered examples with \emph{hard requirements} as these tools do not support \emph{soft requirements}. The on-line version of BoSy tool was used which enforces a maximum timeout of 600 seconds.
For other tools, a local installation on Linux (Ubuntu 16.04) system with Intel i5
64 bit, 2.5 GHz processor and 4 GB memory was used with a time out of 3600 seconds. In this comparison
\DCSYNTH\/ was used with symbolic algorithm for both \MPNC and \GODSC computation. Note that for these examples the \GODSC algorithm will always terminate after 1 iteration only, as the examples
do not have soft requirements, so \DCSYNTH\/ chooses
one of the possible outputs from the \MPNC based on default output order. We have provided default output order for all types of arbiter example as $a_1 > \ldots > a_i$ and for Mine-pump example it is $PumpOn$.

\begin {table}[!h]
\caption {Comparison of Synthesis in Acacia+, BoSy and \DCSYNTH, in terms of controller computation time and memory and number of states of the controller automaton. $Minepump$ as well as $Arb^{tok}(n)$ specifications can be found in Appendix \ref{sec:casestudy_arbiter}.}
\label{tab:comparision}
\begin{minipage}{\textwidth}
        \begin{center}
        \begin{tabular}
        {|c|c|c|c|c|c|c|}
                \hline
                 &    \multicolumn{2}{|c|}{ Acacia+ } &    \multicolumn{2}{|c|}{ BoSy }    & \multicolumn{2}{|c|}{ DCSynth }\\
                \cline{2-7}
                Hard Requirement& time(Sec) & Memory / & time(Sec) & Memory / & time(Sec) & Memory / \\
                & & States & & States & & States \\
                \hline  
                $Arb^{hard}(4,4)$ & 0.4 & 29.8/ 55 &  0.75 & -/4 & 0.08 & 9.1/ 50 \\
                \hline    
                $Arb^{hard}(5,5)$ & 11.4 & 71.9/ 293 & 14.5 & -/8 & 5.03 & 28.1/ 432 \\
                \hline    
                $Arb^{hard}(6,6)$ & TO\footnote{TO=timeout(DCSynth and Acacia+ 3600secs, BoSy 600secs)}  & - & TO  & - & ~ 80 & 1053.0/ 4802\\
                \hline  
                $Arb^{tok}(7)$ & 9.65 & 39.1/ 57 &  TO & - & 0.3 & 7.3/ 7 \\
                \hline    
                $Arb^{tok}(8)$ & 46.44 & 77.9/ 73 & -  & - & 2.2 & 16.2/ 8\\
                \hline    
                $Arb^{tok}(10)$ & NC\footnote{NC=synthesis inconclusive} & - & -  & - & 152 & 82.0/ 10 \\
                \hline    
		Mine-pump & NC & - & TO  & - & 0.06 & 50/ 32\\
		\hline
	\end{tabular}
	Experiments with BoSy are using online version.
\end{center}
\end{minipage}
\end{table}

As the comparison table above
shows, the \DCSYNTH\/  approach seems to outperform the state-of-the-art tools in scalability and 
controller computation time.
This is largely due to the pragmatic design choices made in the logic \qddc\/ and tool \DCSYNTH.

  It can also be seen that BoSy often results in controller with fewer states. BoSy is specifically
optimized to resolve non-determinism to get fewer states. In our case, the tool is optimized
to satisfy maximal number of soft requirements. It would be interesting to merge the two techniques
for best results.

\section{Measuring latency using Model Checking}
\label{sec:qualitymeasure}
Plethora of synthesis algorithms and optimizations give rise to diverse controllers for the same requirement.
In comparing the quality of these different controllers, an important measure is their {\em worst case latency}. Latency can be
defined as time (number of steps) taken to achieve some desired behaviour. In our framework,
for latency specification, user must give a \qddc\/ formula $D^{p}$ characterizing execution fragments of interest. 
For example the \qddc formula $D^{p}$ = \verb#[[req && !ack]]# specifies fragments of execution with 
request continuously {\em true}  but with no acknowledgment. Given a DFA (controller) $M$, the latency goal $MAXLEN(D^{p},M)$
computes $\sup \{e-b \mid \rho,[b,e] \models D^{p}, \rho \in Exec(M) \}$, i.~e.~it computes the length of
the longest interval satisfying $D^{p}$ across all the executions of $M$. Thus, it computes worst case latency for achieving behaviour $D^p$
in $M$. For example, given a synchronous bus arbiter controller $Arb$, 
goal $MAXLEN(\verb#[[req && !ack]]#,Arb)$ specifies the worst case response time of the arbiter $Arb$.
Tool CTLDC, which like \DCSYNTH\/ and DCVALID is member of DCTOOLS suite of tools, provides efficient computation of $MAXLEN$ by 
symbolic search for longest paths as formulated in article \cite{Pan05}. 
This facility will be used subsequently in the paper to compare the  worst case
response times achieved by various controllers synthesized under different criteria.

\oomit{
Given a controlled system $M ~=~ A(Assumptions) \times Controller$ as safety automaton $M$, and
a duration calculus formula $D$, aim is to compute $MAXLEN(D,M)$. We assume that all states of $M$ are reachable.
We first construct NFA $M'$ by turning every accepting state of $M$
as initial state. Clearly $M'$ accepts all fragments of behaviors of $M$. Let $N = M' \times A(D)$. Then,
$N$ accepts all fragments of behaviors of $M$ which satisfy $D$. Now we find the length of the longest accepting path in $N$. Campos \emph{et al} have given symbolic BDD-based algorithm for computing longest and shortest such paths \cite{CCMMH94}. We adapt and implement this in tool DCSynth to compute  the
values of $MAXLEN(D,M)$ as well as $MINLEN(D,M)$. 
This method was previously proposed in \cite{Pan05} where it was shown to be efficient. 
}

\begin {table}[!h]
\caption {Worst Case Response Time Analysis using CTLDC using Response Formula MAXLEN([[$req_i$ \&\& !$ack_i$]]) computation. The value of $H$ is specified
only for $Arb^{soft}$}
\label{tab:perfMeasure2}
\begin{center}
	\begin{tabular}	{|c|c|c|c|c|}
	\hline
		Sr.No & 
		Arbiter Variant & 
		Horizon (H)& 
  \multicolumn{2}{c}{Computed Response} \vline \\
	\hline
	
		& & & 
		Response for $i^{th}$ cell & 
		Value (in cycles) \\		
		
	\hline
	    1 & 
	    $Arb^{hard}(5,5)$  & 
	     -& $1 \leq i \leq 5$ &
	     5 \\
\hline
	    2 & 
	    $Arb^{hardAssume}(5,3,2)$  & 
	    -& i = 1 &
	    2 \\

	    3 & 
	      & 
	     & $2 \leq i \leq 5$ &
	    3 \\
	    
	 \hline	
	    4 & 
	    $Arb^{soft}(5,3)$  & 
	    (H = 1) & $1 \leq i \leq 4$ &
	    $\infty$ \\
	    
	    5 & 
	     & 
	     & i = 5 &
	    3 \\

	 \hline

	    6 & 
	    $Arb^{soft}(5,3)$  & 
	    (H = 2) & $1 \leq i \leq 3$ &
	    $\infty$ \\

	    7 & 
	    & 
	     & $4 \leq i \leq 5$ &
	    3 \\
		
\hline	
		8 & 
		$Arb^{soft}(5,3)$  & 
		(H $>=$ 3) & $1 \leq i \leq 2$ &
		$\infty$ \\

	    9 & 
	    & 
	    & $3 \leq i \leq 5$ &
	    3  \\

	 \hline
	 
\oomit{
		10 & Minepump(MPV1) & [[AssumptionOk \&\& HH2O]] & 5 \\
	 \hline
	    11 & Minepump(MPV2) & [[AssumptionOk \&\& HH2O]] & 7 \\
	 \hline
	    12 & Minepump(MPV3) & [[AssumptionOk \&\& HH2O]] & 8 \\
	 \hline
	 13 & Minepump(MPV4) & [[AssumptionOk \&\& HH2O]] & 6 \\
	 \hline
}
	\end{tabular}
\end{center}
\end{table}

Table.~\ref{tab:perfMeasure2} gives worst case latency measurements carried out using tool CTLDC for various controllers synthesized using \DCSYNTH.
For Arbiter examples, worst case response time (maximum number of cycles a request remains true continuously, without an acknowledgment) is measured using a CTLDC
formula MAXLEN([[$req_i$ \&\& !$ack_i$]]), for each cell $i$ of various arbiters
discussed in section \ref{section:motivation}.
We use arbiter variants with 5 cells (i.e. 
$1 \le i \le 5$) for our experiments.

\oomit{
We also give an example of MinePump case study for different soft requirements and measure its quality
w.r.t. the
maximum time for which the water level can
remain high for different version of Minepump.
The interpretation of the results is described as follows:
}
\begin{itemize}
\item The specification $Arb^{hard}$ and $Arb^{hardAssume}$ do not have soft 
requirements, therefore guided synthesis
will choose an arbitrary output from the 
constructed \MPNC, without any value iteration (i.e. H = 1). The results for these are described as follows
\begin{itemize}

 \item $Arb^{hard}(5,5)$ has worst case response time for each cell as
5 cycles, this would happen when all the 
request lines are continuously $on$ and the
controller gives acknowledgment to each
cell in round robin fashion. 

\item $Arb^{hardAssume}(5,3,2)$ has worst case response for
first cell is 2 cycles, whereas for all the other cells it is 3 cycles, provided the assumptions are met. If assumptions are not
met, then 3 cycle response cannot be guaranteed (If request from all 5 cells is $on$ continuously). 
Assumption put a constraint that at most 2 requests can be $on$ at any
point of time. 

\end{itemize}

\item For the specification $Arb^{soft}(5,3)$ the response requirement is that all the cell should get an
acknowledgment within 3 cycles if the request 
is continuously true (\emph{it would be unrealizable 
if we use only hard
requirement}). However, a controller
which satisfies these requirements as much 
as possible was generated using \DCSYNTH. 
\begin{itemize}

\item For example, $Arb^{soft}(5,3)$ ``tries'' to give acknowledgment within 3 cycles with 
higher priority assigned to higher numbered cell (see the description in Section \ref{section:motivation}). 
However, when all the requests are $on$ simultaneously then $req_5$ gets the highest
priority and hence can always have worst case response time of 3 cycles, but  $req_1$
given the lowest priority may end up with 
worst case response time of $\infty$ (when the request from higher number cell
is always true).  

\item Another important observation is that \DCSYNTH may generates different controllers for different horizons (value iterations)
given for \GODSC computation. More intuitively, as
the value of horizon tends to $\infty$, the 
controller produced reaches closer to the
global optimality.
This effect can be seen from row number 4--9,
where the horizon moves from 1 to more than 2.
For horizon $1$, \DCSYNTH produces a locally
optimal controller and hence the controller
produced only guarantees the response time
for the highest priority cell (i.e. cell no. 5, see row number 5).
For all other cells the worst case response is $\infty$. When the horizon bound is increased to
2, the controller produced meets the response 
requirements for 2 cells (i.e. cell no. 4 and 5, see row number 7). Finally, when the bound is 
increase to 3 or more, the controller produced
guarantees the worst case response for 3 cells (i.e. cell no. 3, 4 and 5 see row number 9).
It can be seen that the maximum number of cells
that can meet the 3 cycle response will be 3, 
in worst case. Therefore, increasing horizon 
beyond 3 does not change the result.

\end{itemize}
\oomit{
\item For Minepump example we synthesized different
controller for different soft requirements.
For the $Minepump$ case study rows 10--13  give the maximum amount of time (in cycles) 
for which the water level can remain high (indicated by the variable {\em HH2O}) 
without violating the assumptions (indicated by {\em AssumptionOk}).
It can be seen that the time for which 
the water remains high continuously shall be
minimum (maximum) when the $pump$ is kept $on$(\emph{off}) continuously,
which can be seen by the result given in row number 10 (12). The details of the Minepump Case Study are given in Appendix \ref{section:casestudy}.
}
\end{itemize}
\vspace*{-0.3cm}

\oomit{
Examples $Arb{soft}$ and $Minepump$ illustrate the impact of soft goals on controller behaviour as well as controller latency under various scenarios and how soft requirements affect the controller synthesis in \DCSYNTH.
}
  
\oomit{
The worst case latency measurement shows that \verb#req6# has response time of 1 cycle 
whereas \verb#req5# has response time of 2 cycles. 
For all other cells  the response time is $\infty$ since cells 6 and 5 can consume all the cycles.
Note that when \verb#req6# is absent throughout 
the response time of cell 4 changes from 
$\infty$ to 3 cycles as shown in the $4^{th}$ row of the Table.~\ref{tab:perfMeasure2}. 
This points to the {\em robustness} of
the synthesized controller.
}

\oomit{
 Here, $Minepump(req)$
denotes $Minepump$ specification with soft requirement $req$ as given in  Appendix \label{section:minepumpexample}.
For example, soft requirement MPV3 is \verb#!PumpOn# which tries to keep pump {\em off} 
as much as possible where as
soft requirement MPV1 is \verb#PumpOn# which tries to keep pump {\em on} as much as possible. 
As a result $Minepump(MPV1)$ gets rid of water in 4 cycles compared to 8 cycles for $Minepump(MPV3)$.}


\end{document}